%% file: main.tex
\documentclass[a4paper,10pt,leqno]{amsart}
\usepackage{amsmath}
\usepackage{amsthm}
\usepackage{amssymb}
\usepackage{stmaryrd}
\usepackage{url}
\usepackage{here}
\usepackage[all]{xy}
\usepackage{comment}
\usepackage{mathtools}
\usepackage{ascmac}
\usepackage{ulem}
\usepackage{framed}
\usepackage{graphicx}
\usepackage{nicematrix}
\usepackage{lmodern}
\usepackage{caption}
\usepackage{subcaption}
\usepackage{footnote}
\usepackage{amscd}
\usepackage{mathdots}
\usepackage{titlesec}
\usepackage{titletoc}
\mathtoolsset{showonlyrefs}

\usepackage{xcolor}
\definecolor{webgreen}{rgb}{0,.5,0}
\definecolor{webbrown}{rgb}{.6,0,0}
\definecolor{RoyalBlue}{cmyk}{1, 0.50, 0, 0}
\usepackage[colorlinks=true, breaklinks=true, urlcolor=webbrown, linkcolor=RoyalBlue, citecolor=webgreen, backref=page]{hyperref}

\titlecontents{section} 
[0pt]                 
{\bfseries}            
{\thecontentslabel\quad} 
{}                    
{\titlerule*[0.5em]{.}\contentspage} 

\titlecontents{subsection} 
[5pt]                 
{}            
{\thecontentslabel\quad} 
{}                    
{\titlerule*[0.5em]{.}\contentspage} 

\titlecontents{subsubsection} 
[10pt]                 
{}            
{\thecontentslabel\quad} 
{}                    
{\titlerule*[0.5em]{.}\contentspage} 

\titlecontents{paragraph} 
[15pt]                 
{}            
{\thecontentslabel\quad} 
{}                    
{\titlerule*[0.5em]{.}\contentspage} 

\titlespacing*{\section}{0pt}{5.5ex plus 1ex minus .2ex}{4.3ex plus .2ex}

\numberwithin{equation}{subsection}

\titleformat{\section}{\centering\normalfont\large\scshape}{\thesection}{1em}{}

\titleformat{\subsection}{\normalfont\normalsize\bfseries}{\thesubsection}{1em}{}

\titleformat{\subsubsection}{\normalfont\normalsize\bfseries}{\thesubsubsection}{1em}{}

\titleformat{\paragraph}{\normalfont\normalsize\it}{\theparagraph}{1em}{}

\theoremstyle{plain}
\newtheorem{thm}{Theorem}[section]
\newtheorem*{thm*}{Theorem}

\newtheorem{prop}[thm]{Proposition}
\newtheorem{cor}{Corollary}
\newtheorem{RHP}{Riemann-Hilbert Problem}

\theoremstyle{definition}

\theoremstyle{remark}
\newtheorem*{rem}{Remark}
\newtheorem*{note}{Note}

\newcommand{\C}{\mathbb C}   
\newcommand{\R}{\mathbb R}

\newcommand{\Z}{\mathbb Z}

\DeclareMathOperator{\res}{res}

\DeclareMathOperator{\diag}{diag}

\newcommand{\su}{{\frak s\frak u}}
\renewcommand{\sl}{{\frak s\frak l}}

\newcommand{\bp}{\begin{pmatrix}} 
\newcommand{\ep}{\end{pmatrix}} 
\renewcommand{\O}{\mathcal{O}}

\newcommand{\e}{\mathbf{e}}

\begin{document}

\title[]{Connection formulae for the radial Toda equations II}

\author[Guest]{Martin A. Guest}
\address{(M.A. Guest) Department of Mathematics, Faculty of Science and Engineering, Waseda University, 3-4-1 Okubo, Shinjuku, Tokyo 169-8555 JAPAN}
\email{martin@waseda.jp}

\author[Its]{Alexander R. Its}
\address{(A.R. Its) Department of Mathematical Sciences, 
Indiana University Indianapolis,
402 N. Blackford St., Indianapolis, IN 46202 USA}
\email{aits@iu.edu}

\author[Kosmakov]{Maksim Kosmakov}
\address{(M. Kosmakov) Department of Mathematical Sciences, University of Cincinnati, P.O. Box 210025, Cincinnati, OH 45221, USA}
\email{kosmakmm@ucmail.uc.edu}

\author[Miyahara]{Kenta Miyahara$^\ast$}
\thanks{$^\ast$ Corresponding author.}
\address{(K. Miyahara) Department of Mathematical Sciences, 
Indiana University Indianapolis,
402 N. Blackford St., Indianapolis, IN 46202 USA}
\email{kemiya@iu.edu}

\author[Odoi]{Ryosuke Odoi}
\address{(R. Odoi) Department of Pure and Applied Mathematics, Faculty of Science and Engineering, Waseda University, 3-4-1 Okubo, Shinjuku, Tokyo 169-8555 JAPAN}
\email{ryosuke.odoi@moegi.waseda.jp}

\begin{abstract}
This is a continuation of \cite{GIKMO} in which we computed the asymptotics near $x = \infty$ of all solutions of the radial Toda equation. 
In this article, we compute the asymptotics near $x = \infty$ of all solutions of a “partner” equation. 
The equations are related in the sense that their respective monodromy data constitute connected components of the same ``monodromy manifold". 
While all solutions of the radial Toda equation are smooth, those of the partner equation have infinitely many singularities, and this makes the Riemann-Hilbert nonlinear steepest descent method (and the asymptotics of solutions) more involved.
\newline

\end{abstract}

\date{\today}

\let\ds\displaystyle

\maketitle 

\setcounter{tocdepth}{3}
\tableofcontents

\setlength{\parskip}{6pt}

\input{section1.tex}

\input{section2.tex}

\input{section3.tex}

\input{Appendix.tex}


\bibliographystyle{alpha}
\bibliography{reference.bib}

\end{document}

%% file: section1.tex
\section{Introduction and Main Results} \label{intro}

In \cite{GIKMO}, we studied the large and small $x$ asymptotic behaviors of the solutions of the radial Toda equation of type $A_2$, namely
\begin{align}
    (w_0)_{xx} + \frac{1}{x} (w_0)_x = 2e^{-2w_0} - 2e^{4w_{0}}, \label{radial Toda with x when n=2}
\end{align}
where $w_0(x)$ is a real-valued function on the positive real line. We performed the desired asymptotic analysis by using the Lax-integrability of equation \eqref{radial Toda with x when n=2}, a fact that goes back to Mikhailov \cite{Mik}:
\begin{align}
\begin{dcases}
    \Psi_{\zeta} = \left(- \frac{1}{\zeta^2} W -\frac{x}{\zeta} w_x - x^2 W^{T} \right) \Psi\\
    \Psi_{x} = \left( -w_x - 2 x \zeta W^T \right) \Psi,
\end{dcases} \label{Lax pair}    
\end{align}
where
\begin{align}\label{W matrix}
    w = \diag (w_0, 0, -w_0) \quad \text{and} \quad W = \begin{pmatrix}
        0 & e^{-w_0} & 0\\
        0 & 0 & e^{-w_0}\\
        e^{2w_0} & 0 & 0
    \end{pmatrix}.
\end{align}
Our study of the radial Toda equation was motivated by the series of works on the tt*-Toda equations by Guest-Its-Lin \cite{GuLi12, GuLi14, GIL1, GIL2, GIL3, GIL4}. The global solutions are physically important \cite{CV1, CV2, CV3} and we have been trying to understand the asymptotic behaviors of more general real-valued radial solutions on $\R_{> 0}$. For both the radial Toda equations and the tt*-Toda equations, we have used a fusion of the Iwasawa factorization method in loop group theory and the Riemann-Hilbert nonlinear steepest descent method of Deift and Zhou \cite{DZ}.

In \cite{GIKMO} we observed that, for a given solution of \eqref{radial Toda with x when n=2}, the corresponding monodromy data for the first linear equation of \eqref{Lax pair} belongs to a $2$-dimensional manifold $\mathcal{M}$.  Explicitly, this manifold can be described by four real parameters $(A^\R, s^\R, x^\R, y^\R)$ satisfying two algebraic equations:
\begin{equation}\label{monodromy cond}
\begin{aligned}
    &(A^{\R})^2 - \frac{1}{3} A^{\R} = (x^\R)^2+(y^\R)^2, \\
    &(1 + s^{\R}) A^{\R} + 2x^\R = \frac{1}{3}.
\end{aligned}
\end{equation}
This has two connected components, given (respectively) by $A^\R > 0$ and $A^\R < 0$. (See Figure \ref{Monodromy manifold for the radial Toda equation with n=2} for a representation of the manifold in $\R^3$.)
\begin{figure}[htbp]
    \centering
    \includegraphics[width=8cm]{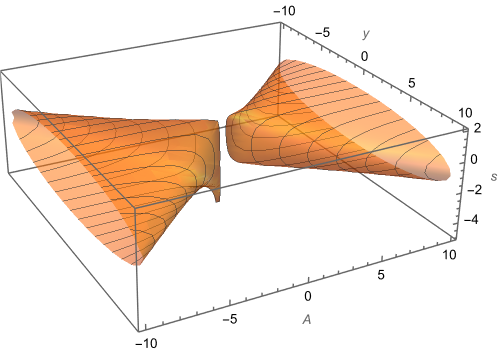}
    \caption{Monodromy manifold $\mathcal{M}$}
    \label{Monodromy manifold for the radial Toda equation with n=2}
\end{figure}

In \cite{GIKMO} we showed that the (real) solutions $w_0(x)$ of \eqref{radial Toda with x when n=2} correspond to the component $A^\R > 0$, and (based on the Riemann-Hilbert method) we observed that points in the component $A^\R < 0$ would arise from the monodromy data of non-real solutions $w_0(x) = v_0(x) + \frac{i \pi}{2}$, $x \in \R_{>0}$, where $v_0(x)$ is real. In other words, the function $v_0(x)$ is a real-valued solution of the following ``new" differential equation
\begin{align}
    (v_0)_{xx} + \frac{1}{x} (v_0)_x = -2e^{-2v_0} - 2e^{4v_0}, \label{third eq}
\end{align}
which could be regarded as a ``monodromy partner" of \eqref{radial Toda with x when n=2}.  The main result of this paper is to prove that the (real) solutions $v_0(x)$ of \eqref{third eq} correspond precisely to the component $A^\R < 0$, by means of a detailed examination of the asymptotics of all solutions near $x = \infty$. Thus the relation between connected components of the manifold $\mathcal{M}$ and equations \eqref{radial Toda with x when n=2} and \eqref{third eq} is not trivial and of some interest in its own right.
\begin{thm} \label{result 1}
The monodromy manifold corresponding to the solutions of \eqref{radial Toda with x when n=2} is the connected component of $\mathcal{M}$ defined by \eqref{monodromy cond} with $A^\R>0$. The monodromy manifold corresponding to the solutions of \eqref{third eq} is the connected component of $\mathcal{M}$ defined by \eqref{monodromy cond} with $A^\R<0$.
\end{thm}


The radial Toda equation \eqref{radial Toda with x when n=2} is known to be equivalent to a special case of the third Painlev\'e equation (PIII) of type $D_7$, see \cite[section 2]{GIKMO}. The same holds for the partner equation, see Corollary \ref{cor 1} (as well for the tt*-Toda equation, in fact). From the Painlev\'e property, it follows that the radial Toda equation, its partner equation (and the tt*-Toda equation) can have only isolated logarithmic poles as ``movable singularities" between zero and infinity. In the case of the radial Toda equation, it is easy to verify that these singularities cannot be real (see \cite[Proposition 2.1]{GIKMO}) by general ODE arguments. In the case of the partner equation, there is no such restriction, indeed our asymptotic analysis shows that they appear in the generic solutions of those equations.



The equation PIII($D_7$) is also known as the degenerate Painlev\'e III equation and has been studied in several contexts. One of the important global aspects is to find the large and small $x$ asymptotics and to connect the asymptotic parameters. In 1987, A. Kitaev \cite{Kitaev} did pioneering work in this direction for a particular choice of parameters, and later Kitaev-Vartanian \cite{KitV1, KitV2} extended the result for arbitrary choice of parameters for purely real solutions and purely imaginary solutions. They used an alternative $2 \times 2$ Lax pair and applied the WKB method to the associated direct monodromy problem.

Although the result by Kitaev-Vartanian already captures the singular solutions of the partner equation \eqref{third eq}, our Riemann-Hilbert approach will give another proof of the large $x$ asymptotics, and along the way prove Theorem \ref{result 1}. 
More importantly, the approach taken in \cite{KitV1, KitV2} does not readily generalize to tackling the asymptotics of the partner equations of the radial Toda equation and the tt*-Toda equation of type $A_n$ for general \(n > 1\), which is our overarching goal. In contrast, our $3 \times 3$ Riemann-Hilbert approach for $n = 2$ is more promising to generalize to the $(n+1) \times (n+1)$ problem for general $n$ in principle. 

It is important to notice that this paper deals with a $3 \times 3$ Riemann-Hilbert problem with singularities which gives extra complication in the Riemann-Hilbert analysis, see section \ref{section 3}. In the $2 \times 2$ situation, there are several prior studies \cite{FIS, BothnerIts}, but this paper will be the first to apply the scheme of those works to the higher rank problems.

Our formulation of the large $x$ asymptotics of \eqref{third eq} is as follows:
\begin{thm}\label{result 2}
For any $(A^{\R}, s^{\R}  ,x^\R,y^\R)$, $A^\R < 0$ satisfying \eqref{monodromy cond}, there is a real solution $v_0(x)$ of 
\begin{align*}
    (v_0)_{xx} + \frac{1}{x} (v_0)_x = -2e^{-2v_0} - 2e^{4v_0},  
\end{align*}
such that 
\begin{align}\label{as th}
    v_0(x) = \frac{1}{2} \ln \left( \frac{1 + \sin \vartheta(x)}{2- \sin \vartheta(x)} \right) + \O(x^{-1}), \quad x \to \infty,
\end{align}
where
\begin{align*}
    &\vartheta(x) = -2\sqrt{3}x + i\nu_0\ln(24\sqrt{3} x) - \frac{\pi}{2} - \arg (x^\R+iy^\R) - \arg \Gamma(-\nu),\\
    &\nu = -\frac{1}{2 \pi i } \ln 3A^\R = -\frac{1}{2} + \nu_0, \quad \nu_0 \in i\R.
\end{align*}
The asymptotics is uniform outside of the small neighborhood of each singularity appearing in \eqref{as th}.
\end{thm}

\begin{rem}
The asymptotic formula \eqref{as th} implies that if we label the singularities of $\frac{1}{2} \ln \left( \frac{1 + \sin \vartheta(x)}{2 - \sin \vartheta(x)} \right)$ by $x_n$ in the increasing order on the real line, then
for $n \in \Z$ large enough, one has
\begin{align}
    x_n = \frac{2 n \pi + \pi/3 + i \nu_0 \ln(24 n \pi) - \arg \left( x^{\R} + iy^{\R} \right) - \arg \left(\Gamma(-\nu)\right)}{2\sqrt{3}} + o(1).
\end{align}
This describes asymptotic locations of movable singularities of $v_0(x)$ for $x$ large. 
\end{rem}

\begin{cor} \label{cor 1}
The change of variables 
\begin{align}
   x = \frac{3}{4}s^{\frac{2}{3}} \quad \text{and} \quad \Tilde{w} = s^{\frac{1}{3}}e^{-2v_0} \label{third eq <--> PIII}
\end{align}
transform \eqref{third eq} into the following $\mathrm{PIII}(D_7)$ equation with $(\alpha, \, \beta, \, \gamma, \, \delta) = (1, 0, 0, 1)$ in the standard notation,
\begin{align}\label{third equation -- Painleve III}
    \Tilde{w}_{ss} = \frac{\Tilde{w}_s^2}{\Tilde{w}} - \frac{\Tilde{w}_s}{s} + \frac{\Tilde{w}^2}{s} + \frac{1}{\Tilde{w}},
\end{align}
and vice versa.
Then, it follows from Theorem \ref{result 2} that the large $x$ asymptotics of a solution to \eqref{third equation -- Painleve III} is given by
the formula
\begin{align*}
    \Tilde{w}(s) = s^{1/3} \, \frac{2 - \sin \Tilde{\vartheta}(s)}{1 + \sin \Tilde{\vartheta}(s)} + \O \left(s^{-2/3}\right), \quad s \to \infty.
\end{align*}
where $\Tilde{\vartheta}(s) = \vartheta \left( \frac{3}{4} s^{2/3} \right)$ and $\vartheta$ is defined in Theorem \ref{result 2}.
\end{cor}


Acknowledgments: The first author was partially supported by JSPS grant 23H00083. The second author was partially supported by NSF grant DMS:1955265, by RSF grant No. 22-11-00070, and by a Visiting Wolfson Research Fellowship from the Royal Society. The fourth author was partially supported by a scholarship from the Japan Student Services Organization and the Hokushin Scholarship Foundation.
Some figures from \cite{Maksim} are reproduced here, with the permission of the author.

%% file: section2.tex
\section{Direct Monodromy Problem}

\subsection{The Partner Equation}

First, we recall from (2.1.3) of \cite{GIKMO} the Lax pair
\begin{align}
\begin{dcases}
    \Psi_{t} = \left( w_t + \frac{1}{\lambda} W \right) \Psi\\
    \Psi_{\bar{t}}= \left( -w_{\bar{t}} - \lambda W^{T} \right) \Psi
\end{dcases} \label{Lax pair, 2dtoda} 
\end{align}
for the Toda equation $2 (w_0)_{t \bar{t}} = e^{-2w_0} - e^{4 w_0}$, where $w_0 : \C \to \R$. The partner equation
\begin{align}
    2(v_0)_{t \bar{t}} = - e^{-2v_0} -  e^{4v_0} \label{partner t}
\end{align}
where $v_0 : \C \to \R$ has an analogous Lax pair
\begin{align}
\begin{dcases}
    \tilde{\Psi}_t = \left(v_t + \frac{1}{\lambda} V \right) \tilde{\Psi}\\
    \tilde{\Psi}_{\bar{t}} = \left( -v_{\bar{t}} - \lambda J V^T J \right) \tilde{\Psi}
\end{dcases} \label{Lax pair for the third eq t and t-bar}
\end{align}
where $v = \diag \left( v_0, 0, -v_0 \right)$, $J = \diag(-1, 1, -1)$, and 
\begin{align} 
\begin{aligned}\label{V and v}
    V = \begin{pmatrix}
           0 & e^{-v_0} & 0\\
           0 & 0 & e^{-v_0}\\
            e^{2v_0} & 0 & 0
         \end{pmatrix}
       = e^{-v} \Pi e^{v}, \quad
       \Pi = \begin{pmatrix}
           0 & 1 & 0\\
           0 & 0 & 1\\
           1 & 0 & 0
       \end{pmatrix}.
\end{aligned}
\end{align}

We remark that \eqref{Lax pair, 2dtoda} and \eqref{Lax pair for the third eq t and t-bar} are associated to different real forms of the Lie algebra $\sl_3 \C$ consisting of $3 \times 3$ complex traceless matrices. 
For the Toda equation, if we restrict $\lambda$ to the circle $|\lambda| = 1$, then the $1$-form $d \Psi \Psi^{-1}$ takes values in the Lie algebra $\su_3$, i.e. the set of matrices $X \in \sl_3 \C$ satisfying $\bar{X}^T = -X$. 
In contrast, for the Lax pair \eqref{Lax pair for the third eq t and t-bar} of the partner equation, the $1$-form $d \tilde{\Psi} \tilde{\Psi}^{-1}$ takes values in the Lie algebra $\su_{2,1}$, given by the condition $J \bar{X}^T J = -X$.


The difference in real forms corresponds to the ``change of variable" $v_0 = w_0 - \frac{i \pi}2$ in the equations for $w_0,v_0$. The Lax pairs are related in the same way. Namely, $e^w = e^v K$ where $K = \diag(i, 1, -i)$, so we have
\begin{align}
    W &= e^{-w} \Pi e^w = K^{-1} e^{-v} \Pi e^v K = K^{-1} V K,\\
    W^T &= 
    K e^{v} \Pi^T e^{-v} K^{-1} =
    K^{-1} J V^T J K.
\end{align}
It follows that \eqref{Lax pair for the third eq t and t-bar} is obtained --- formally --- from \eqref{Lax pair, 2dtoda} by putting
\begin{align}
    \tilde \Psi = K \Psi.
    \label{connecting radial and third Psi}
\end{align}
However, it is important to note that we assume $w_0$ is real in \eqref{Lax pair, 2dtoda}, whereas $v_0$ is real in \eqref{Lax pair for the third eq t and t-bar}, and these cannot hold simultaneously. 
Therefore we interpret the formula $v_0 = w_0 - \frac{i \pi}2$ as a change of \textit{complex} variables, and impose the reality conditions later.

Let us now impose the radial conditions $w_0(t, \bar{t}) = w_0(x)$ and $v_0(t, \bar{t}) = v_0(x)$ where $x = |t|$, so
\begin{align}
    (w_0)_{t \bar t} = \frac{1}{4} \left( (w_0)_{xx} + \frac{1}{x} (w_0)_x \right), \quad (v_0)_{t \bar t} = \frac{1}{4} \left( (v_0)_{xx} + \frac{1}{x} (v_0)_x \right).
\end{align}
For the Toda equation, this gives rise to the isomonodromic Lax pair \eqref{Lax pair}; the analogue of this for the partner equation is 
\begin{align}
\begin{dcases}
    \tilde \Psi_\zeta = \left( -\frac{1}{\zeta^2} V - \frac{x}{ \zeta} v_x - x^2 J V^T J \right) \tilde \Psi \\
    \tilde \Psi_x =  \left( - v_x - 2 x \zeta J V^T J \right) \tilde \Psi
\end{dcases} \label{Lax pair in x}
\end{align}
As in the case of equation (2.2.1) of \cite{GIKMO}, therefore, our first task will be to discuss the monodromy data of the equation
\begin{align}
    \tilde \Psi_\zeta = \left( -\frac{1}{\zeta^2} V - \frac{x}{ \zeta} v_x - x^2 J V^T J \right) \tilde \Psi. 
\end{align}
The simple relation \eqref{connecting radial and third Psi} makes this particularly straightforward. In particular, we may use the same Stokes sectors $\Omega_n^{(0, \infty)}$ where $n = 1, 2$ as defined in \cite{GIKMO}, and obtain analogous formal solutions
\begin{align}
\begin{aligned}
    \tilde \Psi_{f}^{(0)}(\zeta) &=
    e^{-v} \Omega (I + \O(\zeta)) e^{\frac{1}{\zeta} d_3},\\
    \tilde \Psi_{f}^{(\infty)}(\zeta) &= J e^{v} \Omega^{-1} \left(I + \O \left( \zeta^{-1} \right) \right) e^{-x^2 \zeta d_3} 
\end{aligned} \label{formal soln}
\end{align}
where
\begin{align}
    \Omega = \begin{pmatrix}
        1 & 1 & 1\\
        1 & \omega & \omega^2\\
        1 & \omega^2 & \omega
    \end{pmatrix}, \quad d_3 = \diag (1, \omega, \omega^2), \quad \omega = e^{\frac{2 \pi i}{3}}
\end{align}
from $\Psi^{(0)}_f$ and $\Psi^{(\infty)}_f$ in (2.2.2) and (2.2.3) of \cite{GIKMO} by the above change of variable and \eqref{connecting radial and third Psi}.
In each Stokes sector, there is a unique canonical solution  that satisfies
\begin{align}
\begin{aligned}
    &\tilde\Psi_n^{(0)}(\zeta) \sim \tilde\Psi_{f}^{(0)} (\zeta), \quad \zeta \to 0, \quad \zeta \in \Omega_n^{(0)}, \\
    &\tilde \Psi_n^{(\infty)}(\zeta) \sim \tilde \Psi_{f}^{(\infty)} (\zeta), \quad \zeta \to \infty, \quad \zeta \in \Omega_n^{(\infty)}.
\end{aligned}
\end{align}
Then, we define Stokes matrices $\tilde S^{(0)}_n, \tilde S^{(\infty)}_n$ and connection matrices $\tilde E_n$ for \eqref{Lax pair in x}, just as we defined $S^{(0)}_n, S^{(\infty)}_n$ and $E_n$ for \eqref{Lax pair}. 

\begin{prop}
The symmetries satisfied by $\tilde S^{(0)}_n, \tilde S^{(\infty)}_n, \tilde E_n$ are identical to the symmetries satisfied by $S^{(0)}_n, S^{(\infty)}_n, E_n$ (that is, the former symmetries are obtained from the latter symmetries by adding tildes).
\end{prop}

\begin{proof}
We omit the proof, which is entirely analogous to the computations in \cite{GIKMO} of the symmetries satisfied by $S^{(0)}_n, S^{(\infty)}_n, E_n$.
\end{proof}


\subsection{More about the Monodromy Data}

The symmetries of $\tilde S^{(0)}_n, \tilde S^{(\infty)}_n, \tilde E_n$ imply that they can be parametrized by the same monodromy data $m = (A^\R, s^\R, x^\R, y^\R)$ subject to \eqref{monodromy cond} as defined in \cite{GIKMO}, the only difference (as we explain later) being that we have
$A^\R>0$ for the radial Toda equation and $A^\R<0$ for the partner equation. 
In particular,
\begin{align}
     &\tilde S_n^{(\infty)} = \tilde Q_{n}^{(\infty)} \tilde Q_{n+\frac{1}{3}}^{(\infty)} \tilde Q_{n+\frac{2}{3}}^{(\infty)}, \quad
     \tilde S_n^{(0)} =  \tilde Q_n^{(0)} \tilde Q_{n+\frac{1}{3}}^{(0)} \tilde Q_{n+\frac{2}{3}}^{(0)} 
\end{align}
where
\begin{align*}
    &\tilde Q^{(\infty)}_1 = \begin{pmatrix}
        1 & a & 0\\
        0 & 1 & 0\\
        0 & 0 & 1
    \end{pmatrix}, \; \tilde Q_{1\frac{1}{3}}^{(\infty)} = \begin{pmatrix}
        1 & 0 & 0\\
        0 & 1 & 0\\
        0 & -a\omega^2 & 1
    \end{pmatrix}, \; \tilde Q_{1\frac{2}{3}}^{(\infty)} = \begin{pmatrix}
        1 & 0 & 0\\
        0 & 1 & 0\\
        a & 0 & 1
    \end{pmatrix},\\
    &\tilde Q^{(\infty)}_2 = \begin{pmatrix}
        1 & 0 & 0\\
        -a\omega^2 & 1 & 0\\
        0 & 0 & 1
    \end{pmatrix}, \; \tilde Q_{2\frac{1}{3}}^{(\infty)} = \begin{pmatrix}
        1 & 0 & 0\\
        0 & 1 & a\\
        0 & 0 & 1
    \end{pmatrix}, \; \tilde Q_{2\frac{2}{3}}^{(\infty)} = \begin{pmatrix}
        1 & 0 & -a\omega^2\\
        0 & 1 & 0\\
        0 & 0 & 1
    \end{pmatrix},\\
    & \tilde Q^{(0)}_1 = \begin{pmatrix}
        1 & 0 & 0\\
        -a & 1 & 0\\
        0 & 0 & 1
    \end{pmatrix}, \; \tilde Q_{1\frac{1}{3}}^{(0)} = \begin{pmatrix}
        1 & 0 & 0\\
        0 & 1 & a \omega^2\\
        0 & 0 & 1
    \end{pmatrix}, \; \tilde Q_{1\frac{2}{3}}^{(0)} = \begin{pmatrix}
        1 & 0 & -a\\
        0 & 1 & 0\\
        0 & 0 & 1
    \end{pmatrix},\\
    &\tilde Q^{(0)}_2 = \begin{pmatrix}
        1 & a\omega^2 & 0\\
        0 & 1 & 0\\
        0 & 0 & 1
    \end{pmatrix}, \; \tilde Q_{2\frac{1}{3}}^{(0)} = \begin{pmatrix}
        1 & 0 & 0\\
        0 & 1 & 0\\
        0 & -a & 1
    \end{pmatrix}, \; \tilde Q_{2\frac{2}{3}}^{(0)} = \begin{pmatrix}
        1 & 0 & 0\\
        0 & 1 & 0\\
        a\omega^2 & 0 & 1
    \end{pmatrix},
\end{align*}
with $a = \omega^2 s^{\R}$, and
\begin{align}
    \tilde E_1 = \begin{pmatrix}
        A^{\R} & B & \overline{B}\\
        B & \omega s^{\R} A^{\R} - \omega^2 s^{\R} B + \overline{B} & A^{\R}\\
        \overline{B} & A^{\R} & \omega^2 s^{\R} A^{\R} + B - \omega s^{\R} \overline{B}
    \end{pmatrix}, \quad \tilde E_2 = \frac{1}{9}d_3 \tilde E_1^{-1} d_3
\end{align}
where $\det \tilde E_1 = -\frac{1}{27}$ and $\omega B = x^\R + i y^\R$.

We shall see that the sign of $A^\R$ has a significant impact on the implementation of the Riemann-Hilbert analysis of solutions.


Before getting to that point, we will give a detailed parametrization of the monodromy manifold $\mathcal M$ corresponding to $A^\R < 0$. (For a parametrization of the part of $\mathcal M$ corresponding to $A^\R > 0$, see \cite[Proposition 2.14 and the following paragraphs]{GIKMO}.)
From \eqref{monodromy cond},
\begin{align}
    (1 + s^{\R}) A^{\R} + 2 x^{\R} = \frac{1}{3} \iff x^{\R} = \frac{1 - 3(1 + s^{\R}) A^{\R}}{6}, \label{expression for x real}
\end{align}
and
\begin{align}
    (y^{\R})^2 = (A^{\R})^2 - \frac{1}{3} A^{\R} - (x^{\R})^2. \label{identity 8 at parametrization of E_1}
\end{align}
Substituting \eqref{expression for x real} into \eqref{identity 8 at parametrization of E_1} gives
\begin{align}
    (y^{\R})^2 = \frac{(3 + s^{\R})(1 - s^{\R})}{4}(A^{\R})^2 - \frac{1 - s^{\R}}{6} A^{\R} - \frac{1}{36}  \label{identity 9 at parametrization of E_1}
\end{align}
As we are assuming $A^\R < 0$, this formula shows that $s^\R < 1$.
Solving \eqref{identity 9 at parametrization of E_1} for $A^\R$ gives
\begin{align}
    A^{\R} = \frac{(1- s^{\R}) \pm \sqrt{(1-s^\R)^2 + (3 + s^{\R})(1 - s^{\R})(1 + 36 (y^{\R})^2)}}{3(3 + s^{\R})(1-s^\R)}\label{A pm}
\end{align}
when $s^{\R} \neq -3$, and 
\begin{align}
    A^{\R} = -\frac{3(y^\R)^2}{2} -\frac{1}{24}
\end{align}
for $s^\R=-3$.
When $-3<s^\R<1$, only the negative sign in \eqref{A pm}, i.e.,
\begin{align*}
    A^{\R} = \frac{1}{3(3 + s^{\R})} -\frac{2}{3 + s^{\R}} \sqrt{ \frac{1}{36} + \frac{3 + s^{\R}}{1 - s^{\R}} \left( \frac{1}{36} + (y^{\R})^2 \right)}.
\end{align*}
gives $A^\R < 0$ and we have no restrictions on $y^\R$. 
Thus, $(s^\R, y^\R) \in [-3, 1) \times \R$ parametrize the subspace of $\mathcal{M}$ with $A^\R < 0$ given by $s^\R \geq -3$.

For $s^\R<-3$, we need a
different parametrization.
By \eqref{identity 9 at parametrization of E_1},
\begin{align}
   (y^\R)^2+\left(\frac{\sqrt{-(3+s^\R)(1-s^\R)}A^\R}{2} + \sqrt{\frac{1-s^\R}{-36(3+s^\R)}}\right)^2 = \frac{s^\R-1}{36(3+s^\R)}  -\frac{1}{36}
\end{align} 
and one can introduce $\psi^\R \in (0, 2\pi]$ to write
\begin{align}
&y^\R=R \sin \psi^\R,\\
&\frac{\sqrt{-(3+s^\R)(1-s^\R)}A^\R}{2} + \sqrt{\frac{1-s^\R}{-36(3+s^\R)}} = R \cos \psi^\R,\\
&R=\sqrt{ \frac{s^\R-1}{36(3+s^\R)}  -\frac{1}{36} }.
\end{align}
Hence, $(s^\R, \psi^\R) \in (-\infty, -3) \times (0, 2\pi]$ parameterize the subspace of $\mathcal{M}$ with $A^{\R} < 0$ given by $s^{\R} < -3$.

This can be summarized in the following proposition:
\begin{prop}
The points $(A^\R, s^\R, x^\R, y^\R)$ of the connected component of the monodromy manifold $\mathcal{M}$ with $A^{\R} < 0$ can be parametrized as follows. 
\begin{enumerate}
    \item For $s^\R \geq -3$, we have $(s^\R, y^\R) \in [-3, 1) \times \R$ and
    \begin{align*}
      &A^{\R} = \begin{dcases}
       \frac{1}{3(3 + s^{\R})} -\frac{2}{3 + s^{\R}} \sqrt{ \frac{1}{36} + \frac{3 + s^{\R}}{1 - s^{\R}} \left( \frac{1}{36} + (y^{\R})^2 \right)} &s^\R \neq -3,\\
       -\frac{3(y^{\R})^2}{2}-\frac{1}{24} &s^\R = -3,
    \end{dcases}\\
    &x^{\R} = \frac{1 - 3(1 + s^{\R}) A^{\R}}{6}.
    \end{align*}
    \item For $s^\R < -3$, we have $(s^\R, \psi^\R) \in (-\infty, -3) \times (0, 2\pi]$ and
    \begin{align*}
    &A^\R=\frac{2R}{\sqrt{-(3+s^\R)(1-s^\R)}} \cos \psi^\R + \frac{1}{3(3+s^\R)}, \\ &y^\R=R \sin \psi^\R, \quad R=\sqrt{ \dfrac{s^\R-1}{36(3+s^\R)}  -\frac{1}{36} },\\
    &x^{\R} = \frac{1 - 3(1 + s^{\R}) A^{\R}}{6}.
\end{align*}
\end{enumerate}
\end{prop}

%% file: section3.tex
\section{Inverse Monodromy Problem and Asymptotics near $x = \infty$} \label{section 3}
\subsection{Riemann-Hilbert Problems} \label{section 6}

In this subsection, we reconstruct  the solution $v_0(x)$ of \eqref{third eq} from the given monodromy data  
\begin{align*}
    \mathcal{M} \ni m = (A^{\R}, s^{\R}, x^\R, y^\R) \mapsto v_0(x,\, m),
\end{align*}
where $m = (A^{\R}, s^{\R}  ,x^\R,y^\R)$ satisfies \eqref{monodromy cond} and $A^\R < 0$.

This inverse monodromy problem forms a \textit{Riemann-Hilbert Problem} (RHP, in short) for a canonical function $\tilde{\Psi}_n^{(\infty, 0)}$ which is holomorphic on $\Omega_n^{(\infty, 0)} \subset \C^*$, whose jumps are given by either Stokes matrices or connection matrices and whose asymptotics can be given as formal solutions. 
From now on, we shall not use the tilde notation for jump matrices, as we consider only the case of the partner equation.

We first formulate the following $\hat{\Psi}$-Riemann-Hilbert Problem (or $\hat{\Psi}$-problem, in short) exactly in the same way as it is described in \cite{GIKMO}. 
The contour of the $\hat{\Psi}$-problem consists of two vertical rays and a circle $S^1_{\rho} = \{ \zeta \in \C \,|\, |\zeta| = \rho \}$ with $\rho$ arbitrary positive, which we call $\Gamma_1$, and the jump matrices for this problem are shown in Figure \ref{psi-hat problem}. 
From \eqref{formal soln}, the asymptotic behaviors of $\hat{\Psi}(\zeta)$ are
\begin{align}
\hat{\Psi}(\zeta) = 
\begin{dcases}
    J e^{v} \Omega^{-1} (I + \O(\zeta^{-1})) e^{-x^2 \zeta d_3}, & \zeta\rightarrow \infty, \\
    e^{-v} \Omega (I + \O(\zeta)) e^{\frac{1}{\zeta} d_3}, & \zeta\rightarrow 0.
\end{dcases}
\end{align}

\begin{figure}[htbp]
    \centering
    \includegraphics[width=12cm]{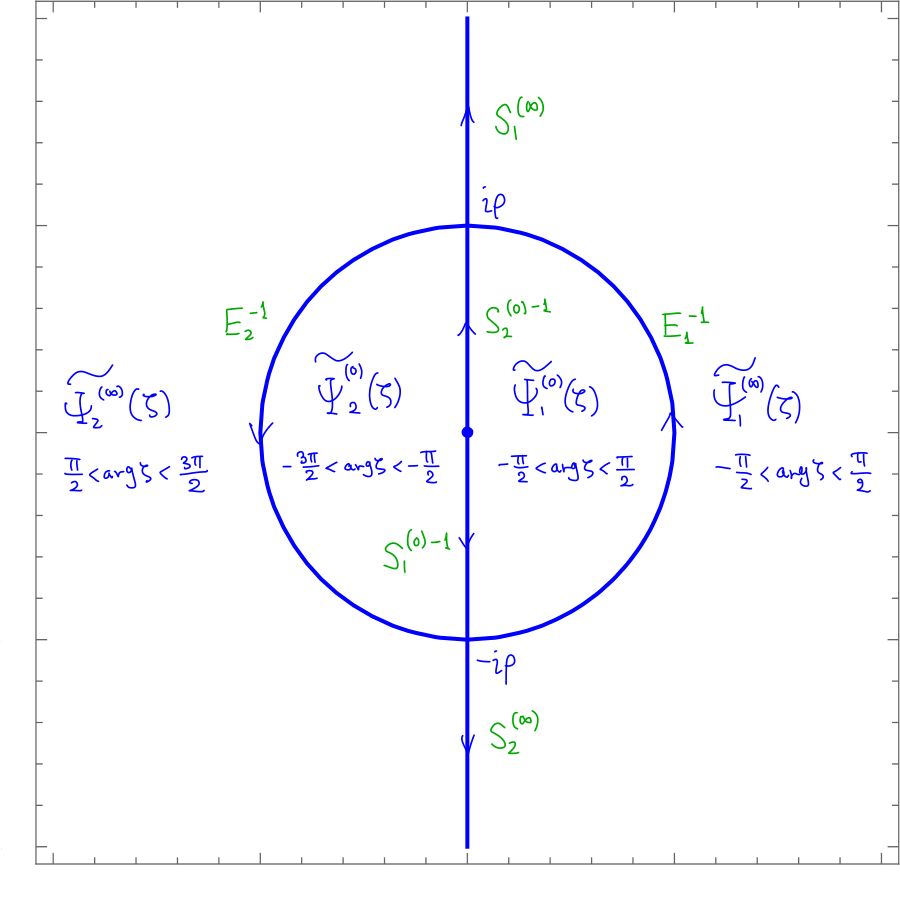}
    \caption{Definition of the $\hat{\Psi}$-problem.}
    \label{psi-hat problem}
\end{figure}


\subsection{Analogies and Differences in the Radial Toda and Partner Equations}

Most calculations for the $\hat{\Psi}$-problem will go in parallel with the ones in \cite{GIKMO}. To avoid repetition, we just sketch the main steps of the coinciding part. Let us briefly indicate them.
 
Applying a series of transformations described in sections 3.1 and 3.2 of \cite{GIKMO} to the $\hat{\Psi}$-problem will give the $\check{Y}$-problem eventually. To solve the $\check{Y}$-problem asymptotically, one can construct globally and locally approximating solutions as in \cite{GIKMO}. In other words, the \textit{global parametrix} $\check{Y}^D(\zeta)$ and the \textit{local parametrices} $P^{(u)}$ near the stationary points $u=\pm 1,\pm \omega,\pm \overline{\omega}$ are given by the same formulae as in the radial Toda case.

The next step is to consider a RHP for an error function $R(\zeta)$  by comparing $\check{Y}(\zeta)$ with its parametrices. This step for the partner equation will completely differ from that for the radial Toda equation.

Recall that for the radial Toda case, we assumed $A^{\R}>0$ (which was consistent with the analysis of the solution near $x = 0$). Under this condition, jump matrices for the error function $R(\zeta)$ tend to the identity matrix as $x\rightarrow\infty$, i.e., the global parametrix and each local parametrix matches well on the boundary of a small neighborhood around each stationary point. Then, we apply the small norm theorem, see \cite[Theorem 8.1]{FIKN}, to the $R$-problem to guarantee the large $x$ solvability of the RHPs. However, in this paper, we assume $A^\R < 0$, and it turns out that we don't have the ``matching up" condition for $R(\zeta)$. In other words, jump matrices $G_R^{(u)}$ on the boundary of each small neighborhood around a stationary point $u$ will not be close to the identity matrix in the $L^2 \cap L^{\infty}$ norm:
\begin{align*}
    & G_R^{(u)}=\check{Y}^{D}(\zeta) \left( P^{(u)}(\zeta) \right)^{-1} = \left[E^{(u)}(\zeta)\right]^{-1} + \O(x^{-1}),\\
    &\bigg|\bigg| \left[E^{(u)}(\zeta) \right]^{-1} - I \bigg|\bigg| \not\to 0, \; \text{ as } \; x \to \infty.
\end{align*}
Because of that, we cannot apply the small norm theorem directly. Such a situation has appeared for $2 \times 2$ matrix RHPs in \cite{BothnerIts} when they considered certain singular solutions of the Painlev\'e II equation. Encountering a non-small norm RHP after standard construction of the global and local parametrices seems to be one feature of studying singular solutions of the corresponding differential equations. In \cite[section 2]{BothnerIts}, a certain ``dressing" procedure was applied to resolve this issue. We follow the same idea with some modifications due to the $3 \times 3$ nature of our problem.


\subsubsection{$\check{Y}$-problem }

As discussed in the previous subsection, the first few steps are to repeat the calculations in \cite{GIKMO} to obtain the $\check{Y}$-problem, whose contour $\Gamma_3$ and the jump matrices $G_{\check{Y}}$ are given by Figure \ref{Y check problem}. See Appendix \ref{E decomp} for explicit expression for $G_{\check{Y}}$.

\begin{framed}
\begin{RHP}\label{YcheckRHP}
Find a matrix-valued function $\check{Y}$ satisfying the following conditions 
\begin{itemize}
    \item $\check{Y}(\zeta) \in H(\C \setminus \Gamma_3)$,  
    \item The jump conditions are
    \begin{align*}
    \check{Y}_{+}(\zeta) &= \check{Y}_{-}(\zeta)   G_{\check{Y}}  ,
    \end{align*}
    \item The normalization condition is 
    \begin{align}
    \check{Y}(\zeta) &= 
    \begin{dcases}
        I + \O(\zeta^{-1}), & \zeta \to \infty,\\
        \Omega e^{-2v} J \Omega^{-1} + \O(\zeta), & \zeta \to 0, \label{Y at 0}
    \end{dcases}
    \end{align}
\end{itemize}
\end{RHP}
\end{framed}

\begin{figure}[htbp]
\centering
\includegraphics[width=12cm]{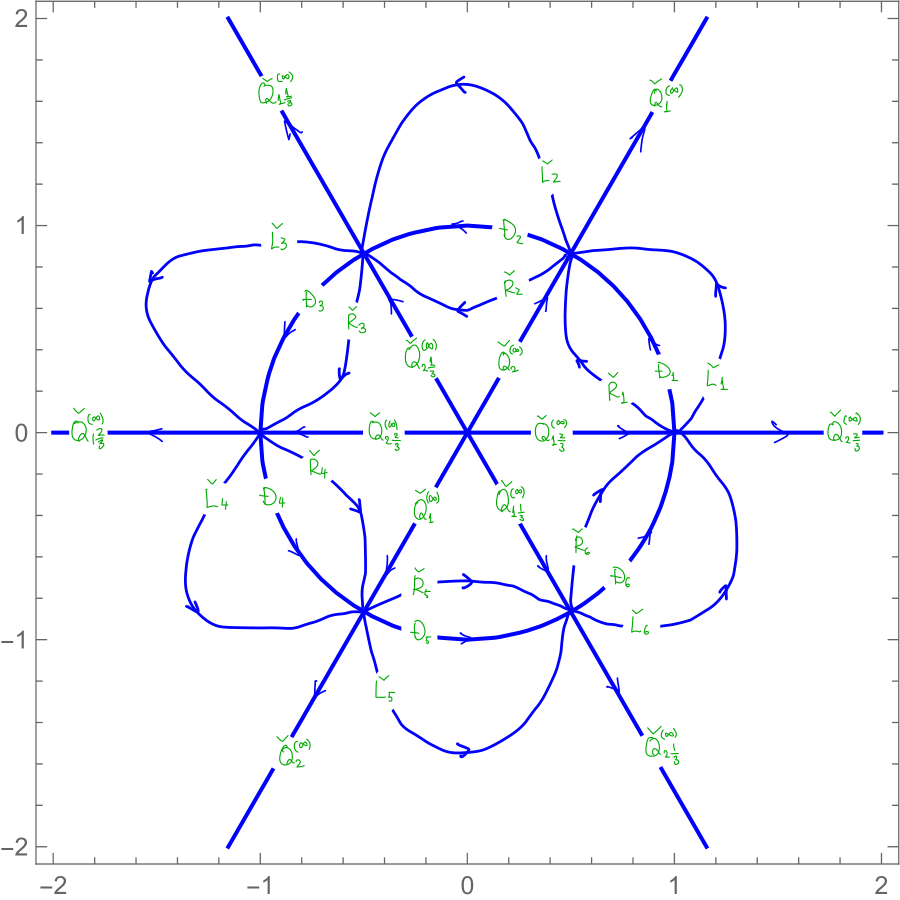}
\caption{Contour $\Gamma_3$ and jump matrices of the $\check{Y}$-problem.}
\label{Y check problem}
\end{figure}


\subsubsection{Global Parametrix} \label{global param}

The global parametrix for the $\check{Y}$-problem is the solution of the following RHP.

\begin{framed}
\begin{RHP}\label{RHP2}
Find a matrix-valued function $\check{Y}^D(\zeta)$ satisfying
\begin{itemize}
    \item $\check{Y}^D(\zeta) \in H(\C \setminus S^{1})$ where $S^1$ is the unit circle. (See Figure \ref{arcs}.)
    \item On the unit circle oriented counterclockwise, the following jump condition is valid:
    \begin{align*}
        \check{Y}^D_+(\zeta) = \check{Y}^D_-(\zeta) G_D,
    \end{align*}
    where the jump matrices are
    \begin{align*}
        G_D = \begin{dcases}
                D_1 & \text{if } \zeta \in C_1\\
                D_2 & \text{if } \zeta \in C_2\\
                D_3 & \text{if } \zeta \in C_3\\
                D_4 & \text{if } \zeta \in C_4\\
                D_5 & \text{if } \zeta \in C_5\\
                D_6 & \text{if } \zeta \in C_6\\
    \end{dcases}
    \end{align*}
    with $D_1, \cdots, D_6$ given in  Appendix \ref{E decomp}.
    \item The normalization condition is $$\check{Y}^D(\zeta) =I+\O(\zeta^{-1})\; \text{ as } \;  \zeta \to \infty.$$
\end{itemize}
\end{RHP}
\end{framed}

\begin{figure}[htbp]
\centering
\includegraphics[width=8cm]{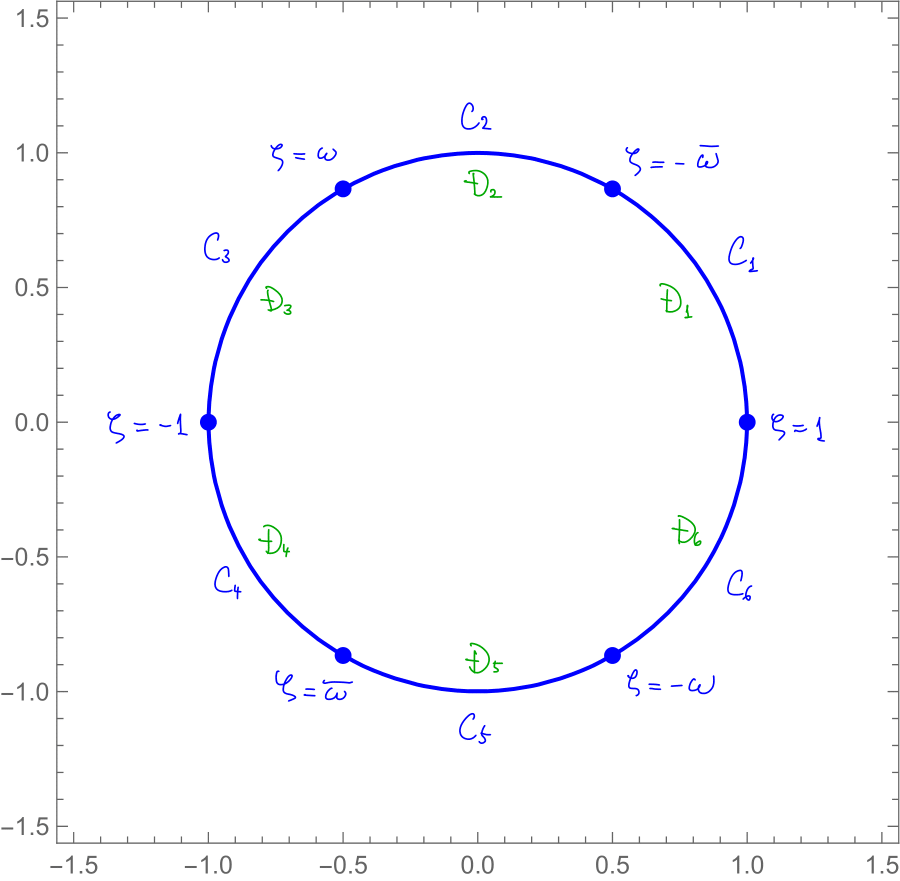}
\caption{Arcs $C_n$, stationary points, and jump matrices of the $\check{Y}^D$-problem.}
\label{arcs}
\end{figure}

The same formula from \cite{GIKMO} gives the solution of $\check{Y}^D$-problem,
\begin{align}
\begin{aligned}
    \check{Y}^{D}(\zeta) &= 
\left( \frac{\zeta -1}{\zeta +\overline{\omega}} \right)^{-\frac{1}{2 \pi i} \ln D_1} \left( \frac{\zeta +\overline{\omega}}{\zeta - {\omega}} \right)^{-\frac{1}{2 \pi i} \ln D_2} \left( \frac{\zeta -{\omega}}{\zeta + 1} \right)^{-\frac{1}{2 \pi i} \ln D_3} \\
    &\times \left( \frac{\zeta + 1}{\zeta -\overline{\omega}} \right)^{-\frac{1}{2 \pi i} \ln D_4} \left( \frac{\zeta -\overline{\omega}}{\zeta + {\omega}} \right)^{-\frac{1}{2 \pi i} \ln D_5} \left( \frac{\zeta +{\omega}}{\zeta - 1} \right)^{-\frac{1}{2 \pi i} \ln D_6}
\end{aligned} \label{global parametrix}
\end{align}
where branches of each power function are fixed by the normalization condition at infinity. 
Set $\nu = -\frac{1}{2 \pi i} \ln 3A^{\R}$. Let $(\zeta - 1)^{\nu}$ be the branch with the cut $\{ \zeta \leq 1\}$ where we take $\arg (\zeta - 1) \in (-\pi, \pi]$. Then, one can write the global parametrix $\check{Y}^D(\zeta)$ in the neighborhood $U_1$ of $\zeta = 1$ as
\begin{align}
\begin{aligned}
\check{Y}^D (\zeta) = \Theta (\zeta)
\begin{pmatrix}
1 & 0 & 0 \\
0 & (\zeta - 1)^{\nu} & 0\\
0 & 0 & (\zeta - 1)^{-\nu}
\end{pmatrix}
\begin{dcases}
    I & \text{for $\zeta \in K^{(1)}$}\\
    D_1 & \text{for $\zeta \in K^{(2)}$}\\
    D_6 & \text{for $\zeta \in K^{(3)}$},
\end{dcases}
\end{aligned} \label{asymptotics of Y^D}
\end{align}
where $\Theta(\zeta) = \diag \left(\Theta_{11}(\zeta), \Theta_{22}(\zeta), \Theta_{33}(\zeta) \right)$ is locally analytic for $\zeta \in U_{1}$ such that
\begin{align}
    \Theta(1) = \begin{pNiceMatrix}[margin]
         e^{\pi i \nu} & & 0  & & & 0  & \\
         0 & & & \Block{2-2}{e^{\frac{\pi i \nu}{2}} \left( 2\sqrt{3}\right)^{-\nu \sigma_3} } \\
         0      
    \end{pNiceMatrix},
\end{align}
and regions $K^{(i)}$ are the following sub-regions of $U_1$:
\begin{align*}
    K^{(1)} &= U_1 \cap \text{(outside of the unit circle)}\\
    K^{(2)} &= U_1 \cap \text{(inside of the unit circle)} \cap \text{(upper half plane)}\\
    K^{(3)} &= U_1 \cap \text{(inside of the unit circle)} \cap \text{(lower half plane)}.
\end{align*}

\begin{figure}[H]
    \centering
    \begin{subfigure}{.33\textwidth}
    \includegraphics[width=0.95\linewidth]{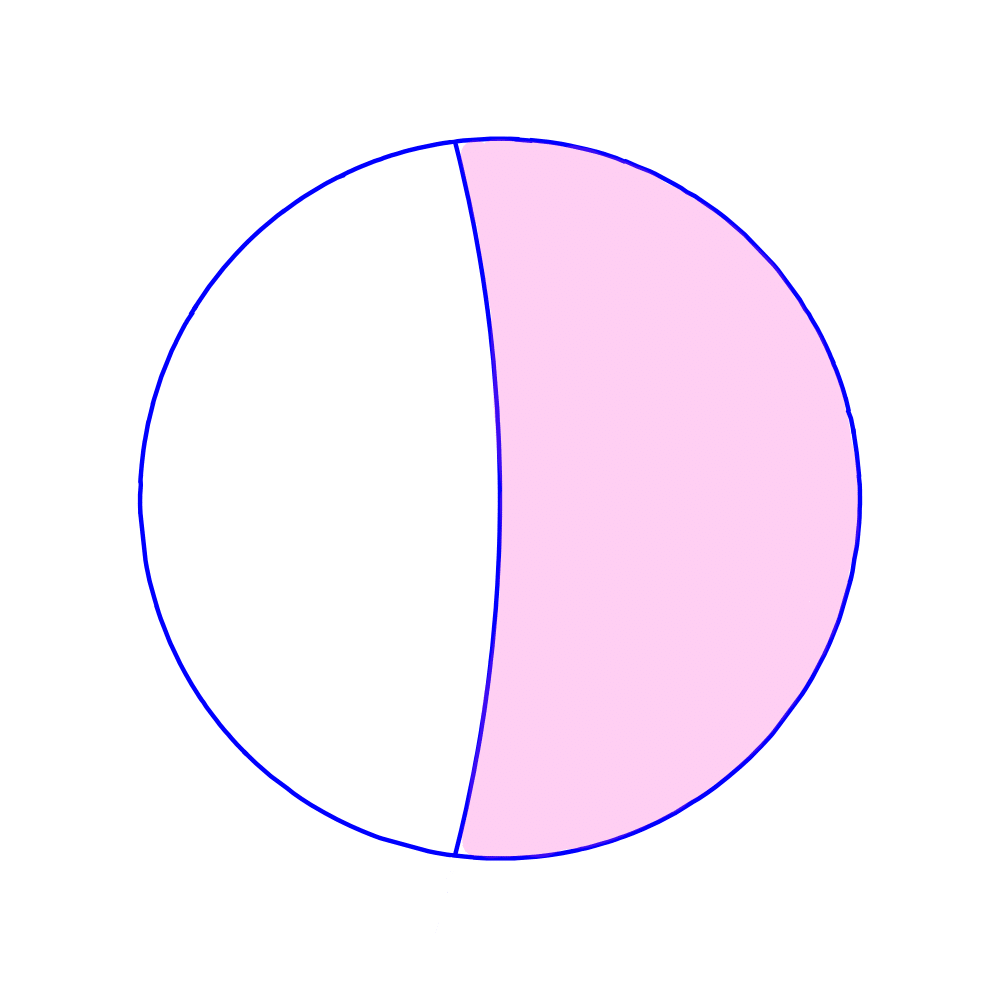}
    \caption{$K^{(1)}$}
    \label{Picture of K^{(1)}}
    \end{subfigure}%
    \begin{subfigure}{.33\textwidth}
    \includegraphics[width=0.95\linewidth]{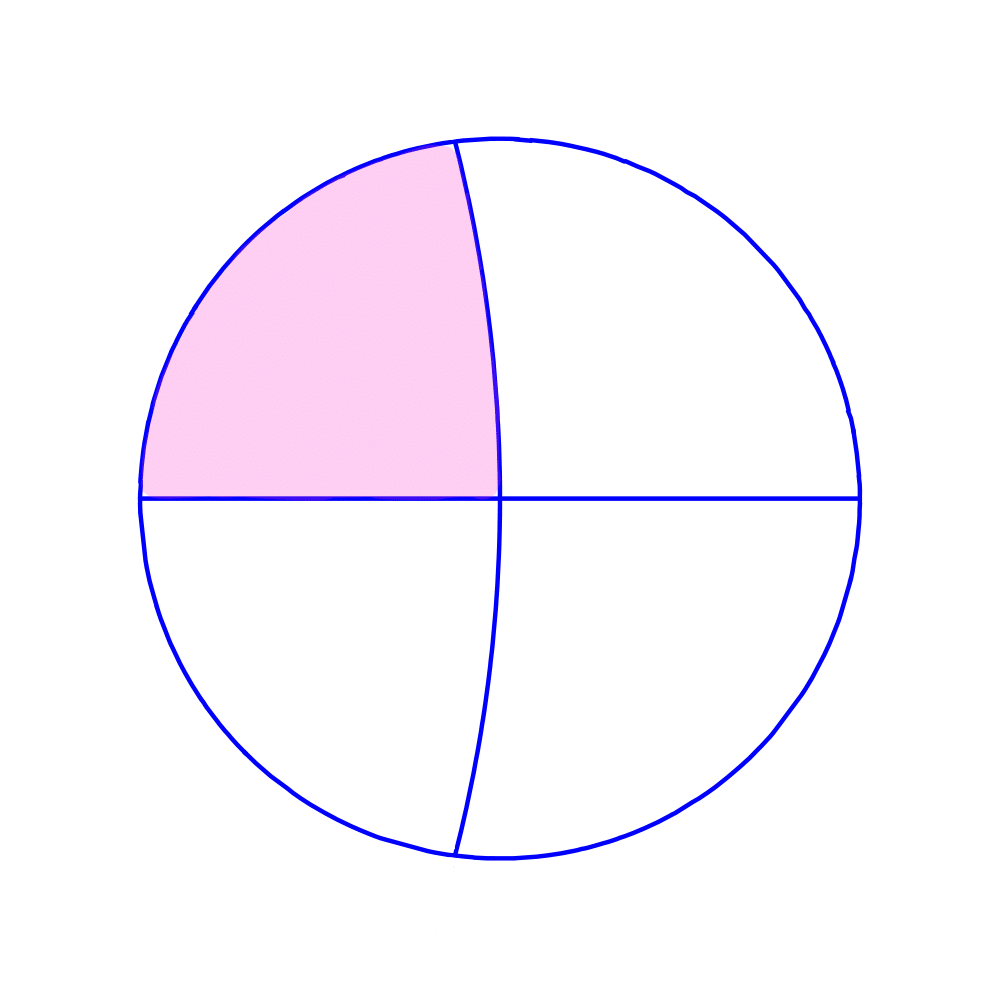}
    \caption{$K^{(2)}$}
    \label{Picture of K^{(2)}}
    \end{subfigure}
    \begin{subfigure}{.33\textwidth}
    \includegraphics[width=0.95\linewidth]{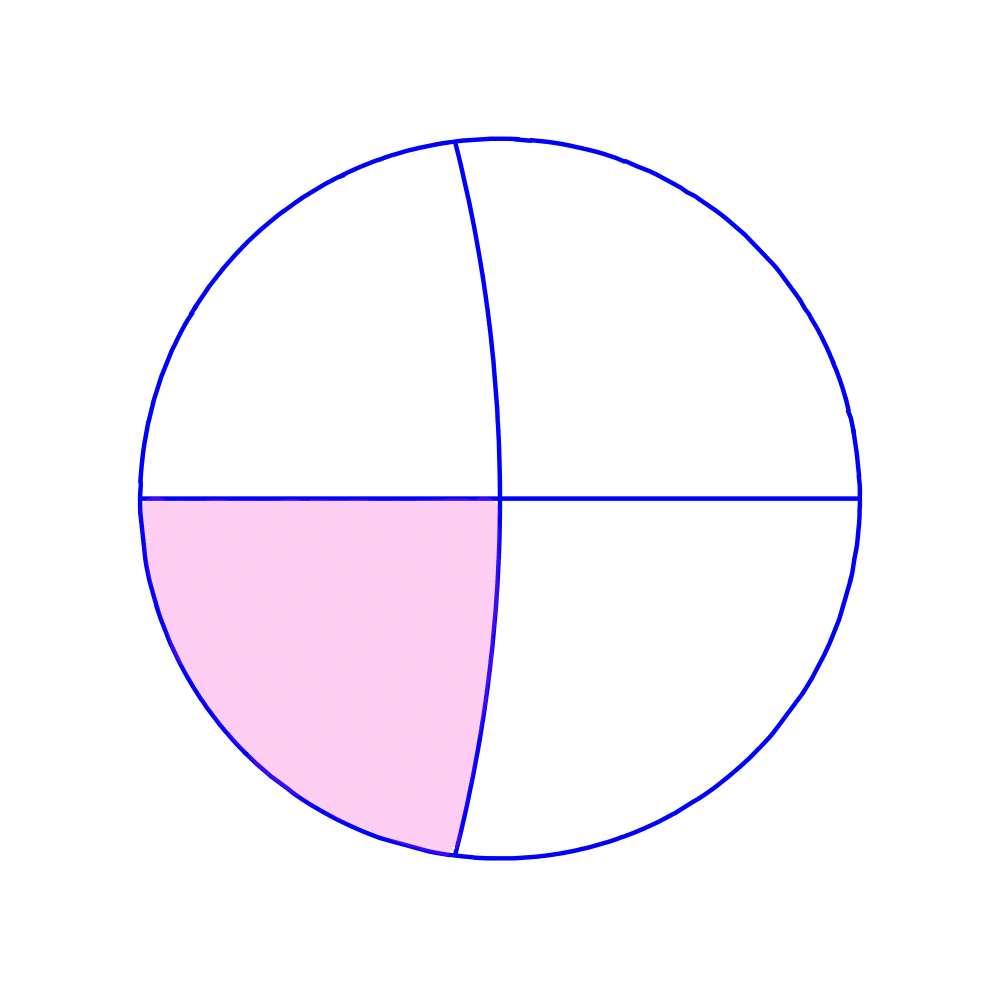}
    \caption{$K^{(3)}$}
    \label{Picture of K^{(3)}}
    \end{subfigure}
    
    \caption{Regions $K^{(i)}$}
    \label{Pictures of K^{(i)}}
\end{figure}


\subsubsection{Local Parametrices }\label{param 1}

The local parametrices also coincide with the ones from \cite[section 3.2.6]{GIKMO}. To construct the local parametrix near $\zeta=1$ for example, we first considered a function $\check{Z}(\zeta)$, which jumps across the curves depicted in Figure \ref{local problem} with the same jump matrices of the $\check{Y}$-problem near $\zeta=1$.
The explicit formula for $\check{Z}(\zeta)$ can be obtained using the parabolic cylinder function. 

\begin{figure}[htbp]
\centering
\includegraphics[width=8cm]{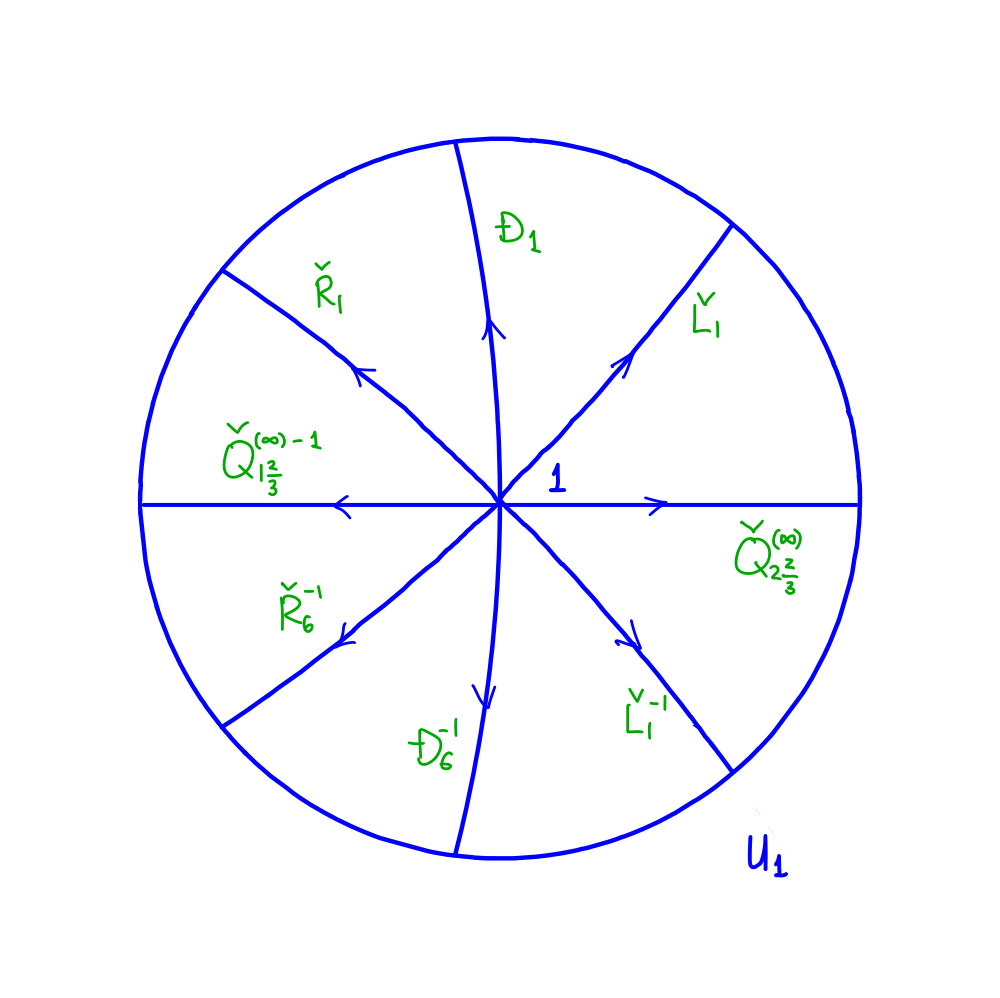}
\caption{Small open disk $U_1$ around $\zeta = 1$, and jump matrices of the $\check{Z}$-problem}
\label{local problem}
\end{figure}

In \cite{GIKMO}, we have shown that $\check{Z}(\zeta)$ has the following asymptotic behavior as $x\rightarrow \infty$,
\begin{align}
\begin{aligned}
\check{Z}(\zeta) = \left( I + \frac{1}{z(\zeta)} \begin{pNiceMatrix}
        0 & 0 & 0 \\
        0 & \Block{2-2}{m} \\
        0 \\
    \end{pNiceMatrix} + \O\left(z^{-2}\right) \right)
\begin{pmatrix}
1 &  &  \\
 & z(\zeta)^{\nu} & \\
 & & z(\zeta)^{-\nu}
\end{pmatrix}
\begin{dcases}
    I & \text{for $\zeta \in K^{(1)}$}\\
    D_1 & \text{for $\zeta \in K^{(2)}$}\\
    D_6 & \text{for $\zeta \in K^{(3)}$},
\end{dcases}
\end{aligned} \label{asymptotics of Phi^0}
\end{align}
 where 
\begin{align}
\begin{aligned}\label{change of var}
    z(\zeta) &= \sqrt{2x} e^{i \frac{\pi}{2}} \sqrt{ i\sqrt3 (\zeta+\zeta^{-1}) - i 2 \sqrt{3}}\\
    &= \sqrt{2x} e^{i \frac{3 \pi}{4}} 3^{1/4} (\zeta - 1) + \O\left( (\zeta - 1)^2 \right),\\
    m &= \begin{pmatrix}
        0 &  -\alpha(x)\\
        \nu/\alpha(x) & 0\\
    \end{pmatrix},\\
    \alpha(x) &= - \frac{i}{s_1 e^{i 2\sqrt{3} x}} \frac{\sqrt{2 \pi} e^{2\pi i \nu}}{\Gamma(-\nu)}, \quad \nu = -\frac{1}{2 \pi i } \ln 3A^\R, \quad s_1 = \omega \frac{x^\R + i y^\R}{A^{\R}}.
\end{aligned}
\end{align}

To make $\check{Z}(\zeta)$ behave like \eqref{asymptotics of Y^D} on $\partial U_1$ to relate the global parametrix, we define the local parametrix $P^{(1)}(\zeta)$ on $U_1$ by
\begin{align}
    P^{(1)}(\zeta) = V(\zeta) \check{Z}(\zeta), \label{defn of P^(1)}
\end{align}
 where $V(\zeta)$ is a holomorphic function in $U_1$ given by
\begin{align} 
    V(\zeta) &= \check{Y}^D(\zeta) \begin{pmatrix}
        1 &  &  \\
         & z(\zeta)^{-\nu} & \\ 
         &  & z(\zeta)^{\nu}
    \end{pmatrix}
    \begin{dcases}
        I & \text{for $\zeta \in K^{(1)}$}\\
        D_1^{-1} & \text{for $\zeta \in K^{(2)}$}\\
        D_6^{-1} & \text{for $\zeta \in K^{(3)}$}.
    \end{dcases} \label{defn of V}
\end{align}
 
Other local parametrices  for  $U_u$, $u = -1, \pm \omega, \pm \overline{\omega}$ are defined by 
\begin{align}
\begin{aligned}
    &P^{(-1)}(\zeta) = d_3^{-1} \left[ P^{(1)}(-\zeta) \right]^{T-1} d_3\\
    &P^{(\omega)} (\zeta) = \Pi P^{(1)}(\overline{\omega} \zeta) \Pi^{-1}\\
    &P^{(-\omega)} (\zeta) = d_3^{-1} \left[ P^{(\omega)}(-\zeta) \right]^{T-1} d_3\\
    &P^{(\overline{\omega})}(\zeta) = d_3^{-1} \left[ P^{(-\overline{\omega})}(-\zeta) \right]^{T-1} d_3\\
    &P^{(-\overline{\omega})} (\zeta) = \Pi^{-1} P^{(-1)}(\omega \zeta) \Pi
\end{aligned} \label{local parametrix near other points}
\end{align}
where
\begin{align*}
    \Pi = \begin{pmatrix}
        0 & 1 & 0\\
        0 & 0 & 1\\
        1 & 0 & 0
    \end{pmatrix}.
\end{align*}

\subsubsection{$R$-problem}
\label{error sec}
Finally, we consider the  $R$-problem defined by
\begin{align}
    R(\zeta) := \check{Y}(\zeta) \left( \check{Y}^{(approx)}(\zeta) \right)^{-1}, \label{definition of R}
\end{align}
where 
\begin{align}
    \check{Y}^{(approx)} (\zeta) = \begin{dcases}
    \check{Y}^D(\zeta) & \text{if } \zeta \in \C \setminus (S^1 \cup U_{all})\\
    P^{(1)}(\zeta) & \text{if } \zeta \in U_1\\
    P^{(- \overline{\omega})}(\zeta) & \text{if } \zeta \in U_{- \overline{\omega}} \\
    P^{(\omega)}(\zeta) & \text{if } \zeta \in U_{\omega} \\
    P^{(-1)}(\zeta) & \text{if } \zeta \in U_{-1} \\
    P^{(\overline{\omega})}(\zeta) & \text{if } \zeta \in U_{\overline{\omega}} \\
    P^{(- \omega)}(\zeta) & \text{if } \zeta \in U_{-\omega}.
    \end{dcases}
\end{align}
 
Observe that the error function $R(\zeta)$ solves the following RHP.

\begin{framed}
\begin{RHP}\label{RHP 4}
Find a matrix-valued function $R( \zeta )$ satisfying the following conditions:
\begin{itemize}
    \item $R( \zeta )$ is holomorphic in $\C \setminus \Gamma_5$ where the contour  $\Gamma_5$ is shown in Figure \ref{contour of R problem}.
    \item $R_+(\zeta ) = R_-( \zeta ) G_{R}$, where the jump matrix $G_{R}$ are
    \begin{align*}
        G_R = \begin{dcases}
            \check{Y}^D(\zeta) G_{\check{Y}} \left( \check{Y}^{D}(\zeta) \right)^{-1} & \text{if $\zeta \in \Gamma_5 \setminus \partial U_{all}$}\\
            \check{Y}^D(\zeta) \left( P^{(u)}(\zeta) \right)^{-1} & \text{if $\zeta \in \partial U_{u}$}.
        \end{dcases}
    \end{align*}
    \item $R(\zeta ) = I + \O(1 / \zeta )$ as $\zeta \to \infty$.
\end{itemize}
\end{RHP}
\end{framed}

\begin{figure}
    \centering
    \includegraphics[width=8cm]{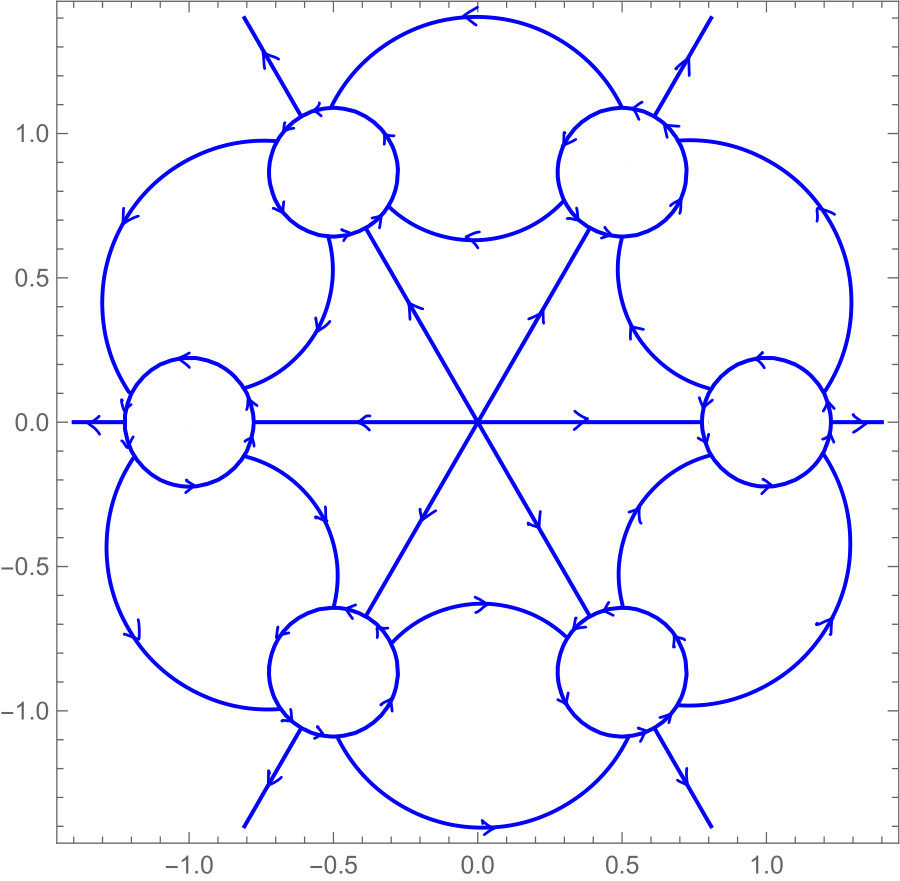}
    \caption{The jump contour $\Gamma_5$ of the $R$-problem.}
    \label{contour of R problem}
\end{figure}

By \eqref{asymptotics of Y^D}, \eqref{asymptotics of Phi^0}, \eqref{defn of P^(1)}, and \eqref{defn of V}, we can further calculate the jump matrix on $\partial U_1$ as follows,
\begin{align} 
    G_R^{(1)}(\zeta) = \check{Y}^D (\zeta) \left[ P^{(1)}(\zeta) \right]^{-1} = \begin{pmatrix}
        1 & 0 & 0 \\
        0 & 1 & \frac{\widehat{\alpha}}{\zeta - 1}\\
        0 & \frac{\widehat{\beta}}{\zeta - 1} & 1
    \end{pmatrix}  +\O(x^{-1})  \label{residue at 1 ver 2}
\end{align}
as $x \to \infty$, where
\begin{align}
\begin{aligned}
    &\widehat{\alpha}(\zeta, x) = \alpha(x) \Theta_{22}(\zeta) \Theta_{33}^{-1}(\zeta) \left( \sqrt{2x} e^{\frac{3\pi i}{4}} 3^{1/4}\right)^{-2\nu-1},\\
    &\widehat{\beta}(\zeta, x) = -\frac{\nu}{\alpha(x)} \Theta_{33}(\zeta) \Theta_{22}^{-1}(\zeta) \left( \sqrt{2x} e^{\frac{3\pi i}{4}} 3^{1/4}\right)^{2\nu-1}.
\end{aligned} \label{alpha-hat, beta-hat}
\end{align}
\begin{note}
For convenience, we omit arguments $\zeta, x$ of \eqref{alpha-hat, beta-hat} hereafter when no confusion is likely.
\end{note}


When $A^\R>0$ as we considered in \cite{GIKMO}, we observe
\begin{equation}
    \nu = -\frac{1}{2\pi i} \ln 3A^\R \in i\R \implies \hat{\alpha}, \, \hat{\beta} = \O \left( x^{-\frac12} \right), \quad x \to \infty,
\end{equation}
uniformly for $\zeta \in \partial U_1$.
This means that \eqref{residue at 1 ver 2} yields the matching relation between parametrices $\check{Y}^D (\zeta)$ and  $P^{(1)}(\zeta)$:
\begin{equation}
   \check{Y}^D (\zeta)  = \left(I+\O \left(x^{-1/2} \right) \right)  P^{(1)}(\zeta),  \quad x \to \infty, \quad \zeta \in U_1,
\end{equation}
which was crucial for the successful implementation of the small norm theorem. This is the reasoning of choosing the left multiplier $V(\zeta)$ in formula \eqref{defn of P^(1)} in the form \eqref{defn of V}.

However, in the case under consideration, we assume $A^\R<0$, which leads to 
\begin{equation}
    \nu = -\frac{1}{2\pi i} \ln 3A^{\R} = -\frac{1}{2} + \nu_0, \quad \nu_0 \in i\R \implies \hat{\alpha} = \O(1), \quad \hat{\beta} = \O(x^{-1}), \quad x \to \infty,
\end{equation}
uniformly for $\zeta \in \partial U_1$.
Thus, instead of the matching relation, we have the formula
\begin{equation}
    \check{Y}^D(\zeta)=\left( \left[E^{(1)}(\zeta)\right]^{-1} + \O \left(x^{-1} \right) \right) P^{(1)}(\zeta), \quad x \to \infty, \quad \zeta \in U_1,
\end{equation}
where 
\begin{align}
    &E^{(1)}(\zeta) = \begin{pmatrix}
        1 & 0 & 0\\
        0 & 1 & -\frac{\widehat{\alpha}}{\zeta - 1}\\
        0 & 0 & 1
    \end{pmatrix} = \O(1). \label{defn of E^{(1)}}
\end{align}
The appearance of the nontrivial matrix $E^{(1)}(\zeta)$ hinders us from using the small norm theorem immediately and results in a further dressing proposed in \cite{BothnerIts}.

We start with the following factorization of the jump matrices $ G_R^{(1)}(\zeta)$:
\begin{align}
    G_R^{(1)}(\zeta)  =   G_{\tilde{R}}^{(1)}(\zeta) \left[E^{(1)}(\zeta)\right]^{-1}
\end{align}
where
\begin{align}
    &G_{\tilde{R}}^{(1)}(\zeta) = \begin{pmatrix}
        1 & 0 & 0\\
        0 & 1 & 0\\
        0 & \frac{\widehat{\beta}}{\zeta - 1} & 1 - \frac{\widehat{\alpha} \widehat{\beta}}{(\zeta - 1)^2}
    \end{pmatrix} + \O(x^{-1}) = I + \O(x^{-1}) \label{jump matrix of R-tilde near 1}
\end{align}
as $x \to \infty$.
Similarly, using the relations \eqref{local parametrix near other points}, one can calculate the jump matrices $G_R^{(u)} (\zeta)$ on $\partial U_u$. They  can be factorized (see Appendix \ref{factorization of the jump matrices}) as follows
\begin{align*}
    G_R^{(u)} (\zeta) = G_{\tilde{R}}^{(u)} (\zeta) \left[ E^{(u)}(\zeta) \right]^{-1}, \quad u = -1, \, \pm\omega, \, \pm \bar{\omega},
\end{align*}
where $G_{\tilde{R}}^{(u)} (\zeta) = I + \O(x^{-1})$ as $x \to \infty$ and 
\begin{align}
\begin{aligned}
    &E^{(-1)}(\zeta) = \begin{pmatrix}
        1 & 0 & 0\\
        0 & 1 & 0\\
        0 & -\frac{\widehat{\alpha} \omega^2}{\zeta + 1} & 1
    \end{pmatrix}, \quad
    E^{(- \overline{\omega})}(\zeta) = \begin{pmatrix}
        1 & 0 & -\frac{\widehat{\alpha} \omega}{\zeta + \overline{\omega}}\\
        0 & 1 & 0\\
        0 & 0 & 1
    \end{pmatrix},\\
    &E^{(\omega)}(\zeta) = \begin{pmatrix}
        1 & -\frac{\widehat{\alpha} \omega}{\zeta - \omega} & 0\\
        0 & 1 & 0\\
        0 & 0 & 1
    \end{pmatrix}, \quad
    E^{(\overline{\omega})}(\zeta) = \begin{pmatrix}
        1 & 0 & 0\\
        0 & 1 & 0\\
        -\frac{\widehat{\alpha} \omega^2}{\zeta - \overline{\omega}} & 0 & 1
    \end{pmatrix},\\
    &E^{(-\omega)}(\zeta) = \begin{pmatrix}
        1 & 0 & 0\\
        -\frac{\widehat{\alpha}}{\zeta + \omega} & 1 & 0\\
        0 & 0 & 1
    \end{pmatrix}.
\end{aligned}
\end{align}

\subsubsection{$\Tilde{R}$-problem}

To eliminate nontrivial matirces $E^{(u)}(\zeta)$, it is natural to pass from $R(\zeta)$ to $\tilde{R}(\zeta)$ defined by
\begin{align}
    \Tilde{R}(\zeta) = 
    \begin{dcases}
        R(\zeta) & \text{if } \zeta \in \C \setminus (S^1 \cup U_{all})\\
        R(\zeta) E^{(1)}(\zeta) & \text{if } \zeta \in U_1\\
        R(\zeta) E^{(- \overline{\omega})}(\zeta) & \text{if } \zeta \in U_{- \overline{\omega}}\\
        R(\zeta) E^{(\omega)}(\zeta) & \text{if } \zeta \in U_{\omega}\\
        R(\zeta) E^{(-1)}(\zeta) & \text{if } \zeta \in U_{-1}\\
        R(\zeta) E^{(\overline{\omega})} (\zeta) & \text{if } \zeta \in U_{\overline{\omega}}\\
        R(\zeta) E^{(-\omega)}(\zeta) & \text{if } \zeta \in U_{-\omega}.
    \end{dcases} \label{R-tilde definition}
\end{align}
It readily follows that all jump matrices of $\tilde{R}(\zeta)$ approach the identity matrix as $x\rightarrow \infty$. However, $\tilde{R}(\zeta)$ will have singularities at $\zeta= \pm1,\pm \omega,\pm\bar{\omega}$. Namely, $\tilde{R}(\zeta)$ solves the following RHP.
 
\begin{framed}
\begin{RHP}\label{RHP 5}
Find a matrix-valued function $\tilde{R}( \zeta )$ satisfying the following conditions:
\begin{itemize}
    \item $\tilde{R}( \zeta )$ is holomorphic in $\C \setminus (\Gamma_5 \cup \{\pm 1, \, \pm \omega, \, \pm \overline{\omega} \})$ where the contour  $\Gamma_5$ is shown in Figure \ref{contour of R problem}.
    \item $\tilde{R}_+(\zeta ) = \tilde{R}_-( \zeta ) G_{R}$, where the jump matrix $G_{\tilde{R}}$ are
    \begin{align}
        \quad G_{\tilde{R}} = \begin{dcases}
            \check{Y}^D(\zeta) G_{\check{Y}} \left( \check{Y}^{D}(\zeta) \right)^{-1} & \text{if $\zeta \in \Gamma_5 \setminus \partial U_{all}$}\\
            G_{\tilde{R}}^{(u)} & \text{if $\zeta \in \partial U_{s},\quad u = \pm1,\pm\omega,\pm\bar{\omega}$}.
        \end{dcases} \label{G_Rtilde}
    \end{align}
    \item $\tilde{R}$ has first-order poles at $\pm 1$, $\pm \omega$, and $\pm \overline{\omega}$. More precisely, when we write $\Tilde{R}(\zeta) = (\Tilde{R}^{(1)}(\zeta), \Tilde{R}^{(2)}(\zeta), \Tilde{R}^{(3)}(\zeta))$ with $\Tilde{R}^{(i)}(\zeta)$ denoting the corresponding  columns of the $3 \times 3$ matrix-valued function, then
    \begin{align}
        &\res_{\zeta = 1} \Tilde{R}^{(3)}(\zeta) = -\widehat{\alpha}(1,x) \Tilde{R}^{(2)}(1) \label{residue relation 1}\\
        &\res_{\zeta = -1} \Tilde{R}^{(2)}(\zeta) = -\widehat{\alpha}(1,x) \omega^2 \Tilde{R}^{(3)}(-1) \label{residue relation 2}\\
        &\res_{\zeta = \omega} \Tilde{R}^{(2)}(\zeta) = -\widehat{\alpha}(1,x) \omega \Tilde{R}^{(1)}(\omega)\label{residue relation 3}\\
        &\res_{\zeta = - \omega} \Tilde{R}^{(1)}(\zeta) = -\widehat{\alpha}(1,x) \Tilde{R}^{(2)}(-\omega)\label{residue relation 4}\\
        &\res_{\zeta = \overline{\omega}} \Tilde{R}^{(1)}(\zeta) = -\widehat{\alpha}(1,x) \omega^2 \Tilde{R}^{(3)}(\overline{\omega}) \label{residue relation 5}\\
        &\res_{\zeta = -\overline{\omega}} \Tilde{R}^{(3)}(\zeta) = -\widehat{\alpha}(1,x) \omega \Tilde{R}^{(1)}(-\overline{\omega}). \label{residue relation 6}
    \end{align}
    \item $\tilde{R}(\zeta ) = I + \O(1 / \zeta )$ as $\zeta \to \infty$.
\end{itemize}
\end{RHP}
\end{framed}

\begin{note}
From \eqref{alpha-hat, beta-hat}, it can be readily verified that
\begin{align}
    &\widehat{\alpha}(1,x) = \frac{\alpha(x)}{\sqrt{2x}} (24 \sqrt{3} x)^{- \nu} 3^{-1/4} e^{-\frac{3 \pi i}{4} - \frac{3 \pi i}{2} \nu} = \O(1), \label{alpha hat (1,x)}\\
    &\widehat{\beta}(1,x) = -\frac{1}{\sqrt{2x}} \frac{\nu}{\alpha(x)} (24 \sqrt{3} x)^{\nu} 3^{-1/4} e^{-\frac{3 \pi i}{4} + \frac{3 \pi i}{2} \nu} = \O(x^{-1})
\end{align}
as $x \to \infty$.
\end{note}

\begin{rem} We first note that the above four conditions of the RHP \ref{RHP 5} determine the solution $\Tilde{R}(\zeta)$ uniquely: indeed relations \eqref{residue relation 1} -- \eqref{residue relation 6} imply 
\begin{equation}\label{representation with sing}
    \tilde{R}(\zeta) = F_{u}(\zeta) E^{(u)}(\zeta), \; \text { for } \zeta \in U_{u}, \; u= \pm1,\pm \omega,\pm\bar{\omega}
\end{equation}
where $F_{u}(\zeta)$ are $3 \times 3$ matrix-valued analytic functions in $U_u$. This implies that one can update the $\Tilde{R}$-problem by defining holomorphic functions $F_u(\zeta)$ in $U_u$ for each $u = \pm 1, \, \pm \omega, \, \pm\overline{\omega}$, and see $\left[E^{(u)}(\zeta)\right]^{-1}$ as extra jump matrices on $\partial U_u$. Thus, one establishes $\det \tilde{R}(\zeta)\equiv 1$ via the Liouville theorem using the normalization at infinity and unimodularity of the jump matrices $G_{\Tilde{R}}$ and $\left[E^{(u)}(\zeta)\right]^{-1}$. From this and representation \eqref{representation with sing}, the ratio of two solutions, $\tilde{R}_1(\zeta)$ and $\tilde{R}_2(\zeta)$ of the $\Tilde{R}$-problem, i.e., $\Tilde{R}_1(\zeta) \left[\tilde{R}_2(\zeta) \right]^{-1}$ is proved to be an entire function approaching identity at infinity; hence $\tilde{R}_1 = \Tilde{R}_2$, showing uniqueness. 

Secondly, this RHP is a small norm problem; however, $\Tilde{R}(\zeta)$ has poles at $\zeta = \pm 1, \pm \omega, \pm \overline{\omega}$. Looking at residue relations, one can deduce the structure of $\Tilde{R}$ as follows:
\begin{align}\label{motiv}
    \tilde{R}=
    \begin{pmatrix}
        \vec{v}_{1}+\dfrac{\vec{w}_1}{\zeta+\omega}+\dfrac{\vec{w}_2}{\zeta-\overline{\omega}}, &\vec{v}_{2}+\dfrac{\vec{q}_1}{\zeta+1}+\dfrac{\vec{q}_2}{\zeta-\omega}, &
        \vec{v}_{3}+\dfrac{\vec{u}_1}{\zeta+\overline{\omega}} +\dfrac{\vec{u}_2}{\zeta-1}
    \end{pmatrix}
\end{align}
where $\vec{v}_{1}, \,  \vec{v}_{2}, \,  \vec{v}_{3}, \, \vec{w}_1, \, \vec{w}_2, \, \vec{q}_1, \, \vec{q}_2, \, \vec{u}_1, \, \vec{u}_2$ are some column vectors.
\end{rem}

\subsubsection{$\hat{R}$-problem}

The second part of the above remark motivates us to define a new function $\hat{R}(\zeta)$ which satisfies a small norm problem but does not have any poles in $\C$, by the following equation,
\begin{align}
    \tilde{R}(\zeta) =(\zeta^2 I &+ \zeta B_1 +  B_2 )\hat{R}(\zeta)\\
    & \times \;\diag \left\{ \dfrac{1}{(\zeta + \omega)(\zeta - \overline{\omega})}, \dfrac{1}{(\zeta - \omega)(\zeta + 1)}, \dfrac{1}{(\zeta + \overline{\omega})(\zeta - 1)} \right\} \label{second idea}
\end{align}
where $B_1$ and $B_2$ are constant matrices in $\zeta$ determined shortly.
Then, $\hat{R}(\zeta)$ solves the following RHP.
\begin{framed}
\begin{RHP}\label{RHP 6}
Find a matrix-valued function $\hat{R}( \zeta )$ satisfying the following conditions:
\begin{itemize}
    \item $\hat{R}( \zeta )$ is holomorphic in $\C \setminus \Gamma_5$ where the contour  $\Gamma_5$ is shown in Figure \ref{contour of R problem}.
    \item $\hat{R}_+(\zeta ) = \hat{R}_-( \zeta ) G_{\hat{R}}$, where the jump matrix $G_{\hat{R}}$ are
    \begin{align*}
        G_{\hat{R}} &= \diag \left\{ \dfrac{1}{(\zeta + \omega)(\zeta - \overline{\omega})}, \dfrac{1}{(\zeta - \omega)(\zeta + 1)}, \dfrac{1}{(\zeta + \overline{\omega})(\zeta - 1)} \right\}\\
        &\qquad \times G_{\tilde{R}} 
        \times \diag \left\{ (\zeta + \omega)(\zeta - \overline{\omega}), (\zeta - \omega)(\zeta + 1), (\zeta + \overline{\omega})(\zeta - 1) \right\}
    \end{align*}
    where $G_{\tilde{R}}$ is given by \eqref{G_Rtilde}.
    \item $\hat{R}(\zeta ) = I + \O(1 / \zeta )$ as $\zeta \to \infty$.
\end{itemize}
\end{RHP}
\end{framed}


The unknown matrices, $B_1$ and $B_2$, are algebraically determined from the conditions \eqref{residue relation 1}, \eqref{residue relation 2}, \eqref{residue relation 3}, \eqref{residue relation 4}, \eqref{residue relation 5}, \eqref{residue relation 6}.
Here we will translate those residue relations into the ones with $B_1$ and $B_2$.

Let $\hat{R}(\zeta) = (\hat{R}^{(1)}(\zeta), \hat{R}^{(2)}(\zeta), \hat{R}^{(3)}(\zeta))$ with $\hat{R}^{(i)}(\zeta)$ denoting its columns. Then, we have
\begin{align}
      \tilde{R}(\zeta) =(\zeta^2 I + \zeta B_1 +  B_2 )\; \left( \dfrac{\hat{R}^{(1)}(\zeta)}{(\zeta + \omega)(\zeta - \overline{\omega})}, \dfrac{\hat{R}^{(2)}(\zeta)}{(\zeta - \omega)(\zeta + 1)}, \dfrac{\hat{R}^{(3)}(\zeta)}{(\zeta + \overline{\omega})(\zeta - 1)} \right).
\end{align}
For instance, the residue relation \eqref{residue relation 1} implies that
\begin{align}
\begin{aligned}
    \qquad &\res_{\zeta = 1} \Tilde{R}^{(3)}(\zeta) = -\widehat{\alpha}(1,x) \Tilde{R}^{(2)}(1)\\
    &\Leftrightarrow (I+B_1+B_2) \hat{R}(1) \begin{pmatrix}
        0\\
        0\\
        \frac{1}{1 + \overline{\omega}}
    \end{pmatrix} = -\widehat{\alpha}(1,x) (I + B_1 + B_2) \hat{R}(1) \begin{pmatrix}
        0\\
        \frac{1}{2(1 - \omega)}\\
        0
    \end{pmatrix}\\
    & \Leftrightarrow (I+B_1+B_2)\dfrac{\hat{R}^{(3)}(1)}{1+ \overline{\omega}} = -\widehat{\alpha}(1,x)(I+B_1+B_2)\dfrac{\hat{R}^{(2)}(1)}{2(1 - \omega) }.
\end{aligned}\label{residue relation I.I}
\end{align}
Similarly, \eqref{residue relation 2} -- \eqref{residue relation 6} imply that
\begin{align}
\begin{aligned}
    \quad &(I-B_1+B_2)\dfrac{\hat{R}^{(2)}(-1)}{1+{\omega}}=\widehat{\alpha}(1,x) \omega^2(I-B_1+B_2)\dfrac{\hat{R}^{(3)}(-1)}{2(1 - \overline{\omega}) }\\
    &(\omega^2I+\omega B_1+B_2)\dfrac{\hat{R}^{(2)}(\omega)}{1+{\omega}} = -\widehat{\alpha}(1,x) \omega (\omega^2I+\omega B_1+B_2)\dfrac{\hat{R}^{(1)}(\omega)}{2(\bar{\omega}-1)}\\
    &(\omega^2I-\omega B_1+B_2)\dfrac{\hat{R}^{(1)}(-\omega)}{\omega+\bar{\omega}} = \widehat{\alpha}(1,x) (\omega^2I-\omega B_1+B_2)\dfrac{\hat{R}^{(2)}(-\omega)}{2(\bar{\omega} - \omega) }\\
    &(\omega I+\bar{\omega} B_1+B_2)\dfrac{\hat{R}^{(1)}(\bar{\omega})}{\omega+\bar{\omega}} = -\widehat{\alpha}(1,x) \omega^2 (\omega I+\bar{\omega} B_1+B_2)\dfrac{\hat{R}^{(3)}(\bar{\omega})}{2( \omega-\bar{\omega}) }\\
    &(\omega I-\bar{\omega} B_1+B_2)\dfrac{\hat{R}^{(3)}(-\bar{\omega})}{ \bar{\omega}+1} = \widehat{\alpha}(1,x) \omega (\omega I-\bar{\omega} B_1+B_2)\dfrac{\hat{R}^{(1)}(-\bar{\omega})}{2(\omega - 1) }.
    \end{aligned}\label{residue relations I.II}
\end{align}


\subsubsection{Asymptotic Analysis} \label{section on asymptotic analysis}

We first note that all jump matrices $G_{\hat{R}}$ are uniformly close to the identity matrix; therefore the $\hat{R}$-problem admits direct asymptotic analysis.
More precisely, we have the following two observations:
\begin{enumerate}
    \item[\textbf{(A)}] On $\Gamma_5 \setminus \partial U_{all}$, jump matrices are exponentially close to the identity matrix.
    \item[\textbf{(B)}] On $\partial U_u$ where $u = \pm 1, \pm \omega, \pm \overline{\omega}$, we have
    \begin{align*}
        ||G_{\hat{R}} - I||_{L^2(\cup \partial U_u) \cap L^{\infty}(\cup \partial U_u)} \leq c_1 x^{-1}
    \end{align*}
    due to \eqref{jump matrix of R-tilde near 1}, \eqref{G^(-omega bar) = E^(-omega bar) G^(-omega bar)}, \eqref{G^(omega) = E^(omega) G^(omega)}, \eqref{G^(-1) = E^(-1) G^(-1)}, \eqref{G^(omega bar) = E^(omega bar) G^(omega bar)}, \eqref{G^(-omega) = E^(-omega) G^(-omega)}.
\end{enumerate}

Then, by the small norm theorem, there is a unique solution $\hat{R}(\zeta)$ of the RHP \ref{RHP 6} for $x$ large.
Indeed, this RHP is equivalent to the singular integral equation,
\begin{align}
    \hat{R}(\zeta) = I + \frac{1}{2 \pi i } \int_{\Gamma_{5}} \frac{\rho(\zeta') (G_{\hat{R}} (\zeta') - I)}{\zeta' - \zeta} d\zeta', \; \zeta \notin \Gamma_{5}, \label{singular integral equation}
\end{align}
where $\displaystyle{\rho(\zeta') = R_{-}(\zeta') = \lim_{\text{\{$-$ side of $\Gamma_5$\}} \ni \lambda \to \zeta'} R(\lambda)}$ for $\zeta' \in \Gamma_5$.
Then, the small norm theorem implies 
\begin{align}
    ||\rho - I||_{L^2(\Gamma_5)} \leq c_2 x^{-1}. \label{rho - I has a small norm}
\end{align}
Taking $\zeta \to 0$ in \eqref{singular integral equation}, we have
\begin{align}
\begin{aligned}
    \hat{R}(0) &= I + \frac{1}{2\pi i} \int_{\Gamma_{5}} \frac{\rho(\zeta')(G_{\hat{R}}(\zeta') - I)}{\zeta'} d\zeta'\\
    &= I + \frac{1}{2\pi i} \int_{\Gamma_{5}} \frac{G_{\hat{R}}(\zeta') - I}{\zeta'} d\zeta' + \frac{1}{2\pi i} \int_{\Gamma_{5}} \frac{(\rho(\zeta') - I)(G_{\hat{R}}(\zeta') - I)}{\zeta'} d\zeta'\\
    &= I + \O (x^{-1})
\end{aligned} \label{R-hat(0) asymptotic equation 1}
\end{align}
where for the last equality we used the above two observations (A) and (B) and the Cauchy-Schwarz inequality.
Similarly, it follows that 
\begin{align}
    \hat{R}(u)=I+\O(x^{-1}) \label{R-hat is I + O(1/x) near singularities}
\end{align}
for $u=\pm1,\pm \omega,\pm\bar{\omega}$. 

When $x$ is sufficiently large, the residue relations, \eqref{residue relation I.I} and \eqref{residue relations I.II}, imply
\begin{align}
\begin{aligned}
    &(I + B_1 + B_2) \frac{\e_3 + \O(x^{-1})}{1 + \omega^2} = - \hat{\alpha}(1,x) (I + B_1 + B_2) \frac{\e_2 + \O(x^{-1})}{2(1 - \omega)},\\
    &(I - B_1 + B_2) \frac{\e_2 + \O(x^{-1})}{1 +\omega} = \hat{\alpha}(1,x) \omega^2 (I - B_1 + B_2) \frac{\e_3 + \O(x^{-1})}{2(1 - \omega^2)},\\
    &(\omega^2 I + \omega B_1 + B_2) \frac{\e_2 + \O(x^{-1})}{1 + \omega} = -\hat{\alpha}(1,x) \omega (\omega^2 I + \omega B_1 + B_2) \frac{\e_1 + \O(x^{-1})}{2(\omega^2 - 1)},\\
    &(\omega^2 I - \omega B_1 + B_2) \frac{\e_1 + \O(x^{-1})}{\omega + \omega^2} = \hat{\alpha}(1,x) (\omega^2 I - \omega B_1 + B_2) \frac{\e_2 + \O(x^{-1})}{2(\omega^2 - \omega)},\\
    &(\omega I + \omega^2 B_1 + B_2) \frac{\e_1 + \O(x^{-1})}{\omega + \omega^2} = -\hat{\alpha}(1,x) \omega^2 (\omega I + \omega^2 B_1 + B_2) \frac{\e_3 + \O(x^{-1})}{2(\omega - \omega^2)},\\
    &(\omega I - \omega^2 B_1 + B_2) \frac{\e_3 + \O(x^{-1})}{1 + \omega^2} = \hat{\alpha}(1,x) \omega (\omega I - \omega^2 B_1 + B_2) \frac{\e_1 + \O(x^{-1})}{2(\omega - 1)},
\end{aligned}
\end{align}
where \(\mathbf{e}_1\), \( \mathbf{e}_2 \), \( \mathbf{e}_3 \) are the standard coordinate vectors in \( \R^3 \).
Thus, we have
\begin{align}
\begin{aligned}
    &(B_1 + B_2) \begin{pmatrix}
        \O(x^{-1})\\
        \frac{\hat{\alpha}(1,x)}{2(1 - \omega)} (1 + \O(x^{-1}))\\
        - \omega^{2} (1 + \O(x^{-1}))
    \end{pmatrix} = - \begin{pmatrix}
        \O(x^{-1})\\
        \frac{\hat{\alpha}(1,x)}{2(1 - \omega)} (1 + \O(x^{-1}))\\
        - \omega^{2} (1 + \O(x^{-1}))
    \end{pmatrix},\\
    &(-B_1 + B_2) \begin{pmatrix}
        \O(x^{-1})\\
        -\omega (1 + \O(x^{-1}))\\
        \frac{\hat{\alpha}(1,x)}{2(1 - \omega)}(1 + \O(x^{-1}))
    \end{pmatrix} = - \begin{pmatrix}
        \O(x^{-1})\\
        -\omega(1 + \O(x^{-1}))\\
        \frac{\hat{\alpha}(1,x)}{2(1 - \omega)}(1 + \O(x^{-1}))
    \end{pmatrix},\\
    &(\omega B_1 + B_2) \begin{pmatrix}
        \frac{\hat{\alpha}(1,x) \omega}{2(-1 + \omega^2)} (1 + \O(x^{-1}))\\
        -\omega (1 + \O(x^{-1}))\\
        \O(x^{-1})
    \end{pmatrix} = - \omega^2 \begin{pmatrix}
        \frac{\hat{\alpha}(1,x) \omega}{2(-1 + \omega^2)} (1 + \O(x^{-1}))\\
        -\omega (1 + \O(x^{-1}))\\
        \O(x^{-1})
    \end{pmatrix},\\
    &(-\omega B_1 + B_2) \begin{pmatrix}
        1 + \O(x^{-1})\\
        \frac{\hat{\alpha}(1,x)}{2(\omega^2 - \omega)} (1 + \O(x^{-1}))\\
        \O(x^{-1})
    \end{pmatrix} = - \omega^2 \begin{pmatrix}
        1 + \O(x^{-1})\\
        \frac{\hat{\alpha}(1,x)}{2(\omega^2 - \omega)} (1 + \O(x^{-1}))\\
        \O(x^{-1})
    \end{pmatrix},\\
    &(\omega^2 B_1 + B_2) \begin{pmatrix}
        -1 + \O(x^{-1})\\
        \O(x^{-1})\\
        \frac{\hat{\alpha}(1,x) \omega^2}{2(\omega - \omega^2)} (1 + \O(x^{-1}))\\
    \end{pmatrix} = - \omega \begin{pmatrix}
        -1 + \O(x^{-1})\\
        \O(x^{-1})\\
        \frac{\hat{\alpha}(1,x) \omega^2}{2(\omega - \omega^2)} (1 + \O(x^{-1}))\\
    \end{pmatrix},\\
    &(-\omega^2 B_1 + B_2) \begin{pmatrix}
        \frac{\hat{\alpha}(1,x) \omega}{2(1 - \omega)} (1 + \O(x^{-1}))\\
        \O(x^{-1})\\
        - \omega^2 (1 + \O(x^{-1}))
    \end{pmatrix} = - \omega\begin{pmatrix}
        \frac{\hat{\alpha}(1,x) \omega}{2(1 - \omega)} (1 + \O(x^{-1}))\\
        \O(x^{-1})\\
        - \omega^2 (1 + \O(x^{-1}))
    \end{pmatrix}.
\end{aligned} \label{residue relations II}
\end{align}

Let
\begin{align*}
    B_1 &= \begin{pmatrix}
        b_{11} & b_{12} & b_{13}\\
        b_{21} & b_{22} & b_{23}\\
        b_{31} & b_{32} & b_{33}
    \end{pmatrix}, \quad 
    B_2 = \begin{pmatrix}
        c_{11} & c_{12} & c_{13}\\
        c_{21} & c_{22} & c_{23}\\
        c_{31} & c_{32} & c_{33}
    \end{pmatrix}, \quad 
    \mathbf{u}_n = \begin{pmatrix}
        b_{n1}\\
        b_{n2}\\
        b_{n3}\\
        c_{n1}\\
        c_{n2}\\
        c_{n3}
    \end{pmatrix}, 
\end{align*}
\begin{align*}
    M &= \begin{pmatrix}
        0 & \frac{\hat{\alpha}(1,x)}{2(1 - \omega)} & -\omega^2 & 0 & \frac{\hat{\alpha}(1,x)}{2(1 - \omega)} & -\omega^2\\
        0 & \omega & -\frac{\hat{\alpha}(1,x)}{2(1 - \omega)} & 0 & - \omega & \frac{\hat{\alpha}(1,x)}{2(1 - \omega)}\\
        \frac{\hat{\alpha}(1,x) \omega^2}{2(-1 + \omega^2)} & -\omega^2 & 0 & \frac{\hat{\alpha}(1,x) \omega}{2(-1 + \omega^2)} &  - \omega & 0\\
        -\omega & \frac{\hat{\alpha}(1,x)}{2(1 - \omega)} & 0 & 1 & \frac{\hat{\alpha}(1,x)}{2(\omega^2 - \omega)} & 0\\
        -\frac{\hat{\alpha}(1,x)}{2(1 - \omega)} & 0 & \omega & \frac{\hat{\alpha}(1,x) \omega}{2(1 - \omega)} & 0 & -\omega^2\\
        - \omega^2 & 0 & \frac{\hat{\alpha}(1,x)}{2(1 - \omega)} & -1 & 0 & \frac{\hat{\alpha}(1,x) \omega}{2(1 - \omega)}
    \end{pmatrix},\\
    \mathbf{v}_1 &= \begin{pmatrix}
        0\\
        0\\
        \frac{\hat{\alpha}(1,x)}{2(1 - \omega^2)}\\
        -\omega^2\\
        -\frac{\hat{\alpha}(1,x) \omega^2}{2(1 - \omega)}\\
        \omega
    \end{pmatrix}, \quad
    \mathbf{v}_2 = \begin{pmatrix}
        \frac{-\hat{\alpha}(1,x)}{2(1 - \omega)}\\
        \omega\\
        1\\
        \frac{\hat{\alpha}(1,x) \omega}{2(1 - \omega)}\\
        0\\
        0
    \end{pmatrix}, \quad
    \mathbf{v}_3 = \begin{pmatrix}
        \omega^2\\
        -\frac{\hat{\alpha}(1,x)}{2(1 - \omega)}\\
        0\\
        0\\
        1\\
        -\frac{\hat{\alpha}(1,x)}{2(\omega - \omega^2)}
    \end{pmatrix}
\end{align*}
for $n = 1, 2, 3$.
The residue conditions \eqref{residue relations II} are then equivalent to the three linear systems 
\begin{align}
    M(I + \O(x^{-1})) \mathbf{u}_n = \mathbf{v}_n + \O(x^{-1})
    \label{residue relations III}
\end{align}
for $n = 1,2,3$. 
We note
\begin{align*}
    \det M = 1 - \frac{\hat{\alpha}(1,x)^2 (648 (1 + i \sqrt{3}) - 384 i \sqrt{3} \hat{\alpha}(1,x) + 54 i (i + \sqrt{3}) \hat{\alpha}(1,x)^2 +
    \hat{\alpha}(1,x)^4)}{1728}.
\end{align*}
Since $\hat{\alpha}(1, x) \neq 0$ by \eqref{alpha hat (1,x)}, it can be readily shown that $\det M = 0$ is equivalent to
\begin{align}
\quad \left( \frac{\sqrt{12}}{\hat{\alpha}(1,x)}\right)^3 - \left( \frac{\hat{\alpha}(1,x)}{\sqrt{12}}\right)^3 - \frac{9}{2} \left( \frac{\sqrt{12}(1 + \sqrt{3}i)}{\hat{\alpha}(1,x)} - \frac{(1 - \sqrt{3}i)\hat{\alpha}(1,x)}{\sqrt{12}} \right) + 16i = 0. \label{det M = 0}
\end{align}
Recall we write $\nu = -\frac{1}{2} + \nu_0$ with $\nu_0 \in i\R$, then one can express $\widehat{\alpha}(1,x)$ as
\begin{align}
    &\widehat{\alpha}(1,x) = \sqrt{12} \, \sigma \left( i e^{ -  i2\sqrt{3}x} (24\sqrt{3}x)^{- \nu_0} \right), \quad
    \sigma =\frac{\sqrt{2 \pi}}{s_1} \frac{ e^{\frac{\pi i \nu_0}{2}}}{\Gamma(-\nu)}.
\end{align}
Since $\Gamma(1+\nu)\Gamma(-\nu) = -\frac{\pi}{\sin \pi \nu}$ and $|s_1|^2=1-e^{2\pi i\nu}=-2e^{\pi i \nu_0} \sin \pi \nu$, it follows that 
\begin{align}
    \bar{\sigma} =  \frac{ s_1}{|{s}_1|^2  } \frac{  \sqrt{2\pi} e^{ \frac12 \pi  {\nu_0}}\Gamma (-\nu)\sin \pi \nu}{-\pi}= \frac{s_1}{\sqrt{2 \pi}} \frac{ \Gamma (-\nu) }{e^{\frac12 \pi i  {\nu_0}}}=\frac{1}{\sigma}. \label{sigma has modulus 1}
\end{align}
By \eqref{change of var} and \eqref{sigma has modulus 1}, one can write $\sigma$ as
\begin{align}
    &\sigma=e^{-i\psi}, \quad \psi = \arg \left( s_1 \Gamma(-\nu) \right)= \frac{2\pi}{3} + \arg \left(x^{\R} + i y^{\R} \right) + \arg \Gamma(-\nu)
\end{align}
and so we have
\begin{align}
    &\widehat{\alpha}(1,x) = \sqrt{12}e^{-i2\sqrt{3}x -\nu_0\ln(24\sqrt{3} x) + i\frac{\pi}{2} - i \psi }=:\sqrt{12}e^{i \phi}. \label{defn of arg = phi}
\end{align}
Thus, \eqref{det M = 0} is equivalent to
\begin{align}
    \left(\sin \left( \phi - \frac{\pi}{3} \right) - 2 \right) \left(\sin \left( \phi - \frac{\pi}{3} \right) + 1 \right)^2 = 0.
\end{align}
This all means that
\begin{align}
    \det M = 0 \iff \phi - \frac{\pi}{3} = -\frac{\pi}{2} - 2 n \pi, \quad n \in \Z. \label{asymp location equation}
\end{align}
We label the $x$ which solves \eqref{asymp location equation} by $x_n$ for $n \in \Z$. Thus, $\{ x_n \}$ is an increasing sequence on $\R_{>0}$ for which $\det M = 0$ holds.
Furthermore,
\begin{align}
    x_n = \frac{(2n+1) \pi + i\nu_0 \ln \left( 24 n \pi \right) - \arg\left( s_1 \Gamma(-\nu) \right) }{2\sqrt{3}} + o(1)
\end{align}
for $n \in \Z$ large enough. 

Let
\begin{align}
    S = \left\{ x \in \R \, | \, \text{$x$ satisfies \eqref{asymp location equation}}\right\}.
\label{defn of S}
\end{align}
Hereafter, we assume $x \in \R$ is sufficiently large and stays away from
\begin{align}
    S_{\epsilon} = \ds \bigcup_{x_n \in S} (x_n - \epsilon, x_n + \epsilon)
\end{align}
for some $\epsilon > 0$. Then, the corresponding $M$ is invertible and \eqref{residue relations III} implies
\begin{align}
    \mathbf{u}_n = M^{-1} \mathbf{v}_n + \O(x^{-1}) \label{u_n = M^-1 v_n + O(1/x)}
\end{align}
for $n = 1, 2, 3$.
If we write the leading term of $\mathbf{u}_n$ by $[\mathbf{u}_n]$,
it follows from \eqref{u_n = M^-1 v_n + O(1/x)} and \eqref{alpha hat (1,x)} that 
\begin{align}
    [\mathbf{u}_n] = M^{-1} \mathbf{v}_n = \O(1), \label{residue relations IV}
\end{align}
In other words, $B_1$ and $B_2$ can be constructed uniquely from \eqref{residue relations III} for $x \in \R_{>0} \setminus S_{\epsilon}$ sufficiently large and have
\begin{align}
    B_1 = [B_1] + \O(x^{-1}), \quad B_2 = [B_2] + \O(x^{-1}), \label{B_1 and B_2 are O(1)}
\end{align}
where $[B_1] = \left( [b_{nm}] \right)_{1 \leq n, m \leq 3}$ and $[B_2] = \left( [c_{nm}] \right)_{1 \leq n, m \leq 3}$ are non-decaying leading terms of $B_1$ and $B_2$ respectively where entries $[b_{nm}]$ and $[c_{nm}]$ are those of $[\mathbf{u}_n]$.

To find leading terms of such $B_1$ and $B_2$, one can perform direct computations \eqref{residue relations IV}. For instance, one has
\begin{align}
    - \omega [c_{33}] &= 1 - \frac{-6 \left( \frac{\sqrt{12}(1 + \sqrt{3}i)}{\hat{\alpha}(1,x)} - \frac{(1 - \sqrt{3}i)\hat{\alpha}(1,x)}{\sqrt{12}}\right) + 24 i}{\left( \frac{\sqrt{12}}{\hat{\alpha}(1,x)}\right)^3 - \left( \frac{\hat{\alpha}(1,x)}{\sqrt{12}}\right)^3 - \frac{9}{2} \left( \frac{\sqrt{12}(1 + \sqrt{3}i)}{\hat{\alpha}(1,x)} - \frac{(1 - \sqrt{3}i)\hat{\alpha}(1,x)}{\sqrt{12}} \right) + 16i}\\
    &=: 1 - X. \label{-omega c33 = 1 -X}
\end{align}
Moreover, it can be observed that
\begin{align}
    &- [c_{11}] = -\omega^2 [c_{22}] = -\omega [c_{33}],\\
    &- [c_{21}] = -\omega^2 [c_{32}] = -\omega [c_{13}],\\
    &- [c_{31}] = -\omega^2 [c_{12}] = -\omega [c_{23}].
\end{align}

By definition \eqref{second idea} of $\hat{R}$ and definition \eqref{R-tilde definition} of $\tilde{R}$, when $\zeta$ belongs to the interior of $S^1$,  
\begin{align*}
    R(\zeta) = (\zeta^2 I &+ \zeta B_1 +  B_2)\; \hat{R}(\zeta)\\
    &\times \diag \left\{ \frac{1}{(\zeta + \omega)(\zeta - \overline{\omega})}, \frac{1}{(\zeta - \omega)(\zeta + 1)}, \frac{1}{(\zeta + \overline{\omega})(\zeta - 1)} \right\}.
\end{align*}
From \eqref{R-hat(0) asymptotic equation 1} and \eqref{B_1 and B_2 are O(1)}, we have
\begin{align}
\begin{aligned}
    \lim_{\zeta \to 0} R(\zeta) = &[B_2] \begin{pmatrix}
        -1 & & \\
        & -\omega^2 & \\
        & & - \omega
    \end{pmatrix} + \O(x^{-1})\\
    &= \begin{pmatrix}
        -[c_{11}] & -\omega^2 [c_{12}] & -\omega [c_{13}]\\
        -[c_{21}] & -\omega^2 [c_{22}] & -\omega [c_{23}]\\
        -[c_{31}] & -\omega^2 [c_{32}] & -\omega [c_{33}]
    \end{pmatrix} + \O(x^{-1}).
\end{aligned} \label{R(0) asymptotic 1}
\end{align}

On the other hand, we have
\begin{align}
    \lim_{\zeta\rightarrow 0}   \check{Y} (\zeta) =\lim_{\zeta\rightarrow 0}   R(\zeta) \check{Y}^{D} (\zeta) \label{Y(0) = R(0) Y^D(0)}
\end{align}
by definition \eqref{definition of R} of $R(\zeta)$.
From \eqref{global parametrix} and the choice of branches (see \cite[Theorem 3.2]{GIKMO}), one can see that $\check{Y}^D(0) = I$. Together with \eqref{Y at 0}, we have
\begin{align}\begin{aligned}
    \qquad &\lim_{\zeta\rightarrow 0} {R}(\zeta)  =\lim_{\zeta\rightarrow 0} \check{Y}(\zeta)= \Omega e^{-2v} J \Omega^{-1}  \\
    &= \frac{1}{3} \begin{pmatrix}
        1 - e^{2 v_0} - e^{-2 v_0} & \omega^2 - \omega e^{2 v_0} - e^{-2 v_0} & \omega - \omega^2 e^{2 v_0} - e^{-2 v_0} \\
        \omega - \omega^2 e^{2 v_0} - e^{-2 v_0} & 1 - e^{2 v_0} - e^{-2 v_0} & \omega^2 - \omega e^{2 v_0} - e^{-2 v_0}\\
        \omega^2 - \omega e^{2 v_0} - e^{-2 v_0} & \omega - \omega^2 e^{2 v_0} - e^{-2 v_0} & 1 - e^{2 v_0} - e^{-2 v_0}
    \end{pmatrix}. \label{R(0) asymptotic 2}
\end{aligned}\end{align}

By \eqref{defn of arg = phi}, \eqref{-omega c33 = 1 -X}, \eqref{R(0) asymptotic 1}, and \eqref{R(0) asymptotic 2}, one has
\begin{align}
    e^{2v_0} + e^{-2v_0} = 3X - 2 + \O(x^{-1}) \label{e^{2v} + e^{-2v} = 3X - 2}
\end{align}
where
\begin{align*}
    X &= \frac{3}{(2 - \sin (\phi - \frac{\pi}{3}))(1 + \sin (\phi - \frac{\pi}{3}))}.
\end{align*}

\begin{rem}
One can compute $\min X = \frac43$ which implies $3X-2 \geq 2$. This is consistent with the fact $e^{2v_0}+e^{-2v_0} \geq 2$.
\end{rem}

Solving \eqref{e^{2v} + e^{-2v} = 3X - 2}, one has
\begin{align}
    e^{-2v_0(x)} = \frac{2- \sin(\phi - \frac{\pi}{3})}{1 + \sin(\phi - \frac{\pi}{3})} + \O \left(x^{-1}\right), \label{e^-2v = fraction of sin}
\end{align}
which is consistent with comparison of all the entries of \eqref{R(0) asymptotic 1} with those of \eqref{R(0) asymptotic 2}. In other words,
\begin{align}
    v_0(x) = \frac{1}{2} \ln \left( \frac{1 + \sin(\phi - \frac{\pi}{3})}{2- \sin(\phi - \frac{\pi}{3})} \right) + \O\left(x^{-1}\right). \label{v_0 = 1/2 log (...)}
\end{align}
Writing $\vartheta = \phi - \frac{\pi}{3}$ completes a proof of Theorem \ref{result 2}.

\begin{rem}
By \eqref{v_0 = 1/2 log (...)}, the locations of the singularities of $v_0(x)$ are asymptotically given by the locations of singularities of $\frac{1}{2} \ln \left( \frac{1 + \sin(\phi - \frac{\pi}{3})}{2- \sin(\phi - \frac{\pi}{3})} \right)$ which are $x_n \in S$.
\end{rem}

Our asymptotic result also shows that there is no solution of the radial Toda equation with $A^\R < 0$; simultaneously, the result of \cite{GIKMO} shows that there is no solution of the partner equation with  $A^\R > 0$. Hence we have Theorem \ref{result 1}.

%% file: Appendix.tex
\appendix

\section{Jump Matrices for the $\check{Y}$-problem}\label{E decomp}

Let $\theta (\zeta) = -\zeta d_3 + \frac{1}{\zeta} d_3^{-1}$ and then we write 
\begin{align}\label{LLcheck}
    \check{L}_k = e^{x \theta(\zeta)} L_k e^{-x \theta(\zeta)}, \quad
    \check{R}_k = e^{x \theta(\zeta)} R_k e^{-x \theta(\zeta)},
\end{align}
where
\begin{align}
\begin{aligned}
    L_1 &= \begin{pmatrix}
     1 & - \frac{B}{A^{\R}} - \omega^2 \frac{|B|^2}{(A^{\R})^2} & -\frac{\overline{B}}{A^{\R}}\\
     0 & 1 & 0\\
     0 & \omega^2 \frac{B}{A^{\R}} & 1
    \end{pmatrix}, \;
    D_1 = \begin{pmatrix}
        \frac{1}{3A^{\R}} & & \\
        & 3 A^{\R} & \\
        & & 1
    \end{pmatrix},\\
    R_1 &= \begin{pmatrix}
        1 & 0 & 0\\
        - \frac{\overline{B}}{A^{\R}} - \omega \frac{|B|^2}{(A^{\R})^2} & 1 & \omega\frac{\overline{B}}{A^{\R}}\\
        - \frac{B}{A^{\R}} & 0 & 1
    \end{pmatrix}.
\end{aligned} \label{parametrization of L_1 D_1 R_1 ver.2}
\end{align}
\begin{align}
\begin{aligned}
    L_2 &= \begin{pmatrix}
     1 & \omega\frac{\overline{B}}{A^{\R}} & 0\\
     0 & 1 & 0\\
     - \frac{B}{A^{\R}} & - \omega \frac{|B|^2}{(A^{\R})^2} - \omega s^{\R} - \frac{\overline{B}}{A^{\R}} & 1
    \end{pmatrix}, \;
    D_2 = \begin{pmatrix}
       1 & & \\
        & 3 A^{\R} & \\
        & & \frac{1}{3A^{\R}}
    \end{pmatrix},\\
    R_2 &= \begin{pmatrix}
        1 & 0 & - \frac{\overline{B}}{A^{\R}}\\
        \omega^2 \frac{B}{A^{\R}} & 1 & - \omega^2 \frac{|B|^2}{(A^{\R})^2} - \omega^2 s^{\R} - \frac{B}{A^{\R}}\\
        0 & 0 & 1
    \end{pmatrix}.
\end{aligned} \label{parametrization of L_2 D_2 R_2 ver.2}
\end{align}
\begin{align}
\begin{aligned}
    L_3 &= \begin{pmatrix}
     1 & 0 & 0\\
     \omega^2 \frac{B}{A^{\R}} & 1 & 0\\
     -\omega^2 \frac{|B|^2}{(A^{\R})^2} - \frac{B}{A^{\R}} & - \frac{\overline{B}}{A^{\R}} & 1
    \end{pmatrix}, \; 
    D_3 = \begin{pmatrix}
       3 A^{\R} & & \\
        & 1 & \\
        & & \frac{1}{3A^{\R}}
    \end{pmatrix},\\
    R_3 &= \begin{pmatrix}
        1 & \omega \frac{\overline{B}}{A^{\R}} & -\omega \frac{|B|^2}{(A^{\R})^2} - \frac{\overline{B}}{A^{\R}}\\
        0 & 1 & -\frac{B}{A^{\R}}\\
        0 & 0 & 1
    \end{pmatrix}.
\end{aligned} \label{parametrization of L_3 D_3 R_3 ver.2}
\end{align}
\begin{align}
\begin{aligned}
    L_4 &= \begin{pmatrix}
     1 & 0 & 0\\
    -\omega \frac{|B|^2}{(A^{\R})^2} - \omega s^{\R} - \frac{\overline{B}}{A^{\R}} & 1 & - \frac{B}{A^{\R}}\\
     \omega \frac{\overline{B}}{A^{\R}} & 0 & 1
    \end{pmatrix}, \;
    D_4 = \begin{pmatrix}
       3 A^{\R} & & \\
        & \frac{1}{3 A^{\R}} & \\
        & & 1
    \end{pmatrix},\\
    R_4 &= \begin{pmatrix}
        1 & -\omega^2 \frac{|B|^2}{(A^{\R})^2} - \omega^2 s^{\R} - \frac{B}{A^{\R}} & \omega^2 \frac{B}{A^{\R}}\\
        0 & 1 & 0\\
        0 & -\frac{\overline{B}}{A^{\R}} & 1
    \end{pmatrix}.
\end{aligned} \label{parametrization of L_4 D_4 R_4 ver.2}
\end{align}
\begin{align}
\begin{aligned}
    L_5 &= \begin{pmatrix}
     1 & 0 & \omega^2 \frac{B}{A^{\R}}\\
     -\frac{\overline{B}}{A^{\R}} & 1 & - \omega^2 \frac{|B|^2}{(A^{\R})^2} - \frac{B}{A^{\R}}\\
     0 & 0 & 1
    \end{pmatrix}, \; 
    D_5 = \begin{pmatrix}
       1 & & \\
        & \frac{1}{3A^{\R}} & \\
        & & 3 A^{\R}
    \end{pmatrix},\\
    R_5 &= \begin{pmatrix}
        1 & -\frac{B}{A^{\R}} & 0\\
        0 & 1 & 0\\
        \omega \frac{\overline{B}}{A} & - \omega \frac{|B|^2}{(A^{\R})^2} - \frac{\overline{B}}{A^{\R}} & 1
    \end{pmatrix}.
\end{aligned} \label{parametrization of L_5 D_5 R_5 ver.2}
\end{align}
\begin{align}
\begin{aligned}
    L_6 &= \begin{pmatrix}
     1 & - \frac{B}{A^{\R}} & -\omega \frac{|B|^2}{(A^{\R})^2} - \omega s^{\R} - \frac{\overline{B}}{A^{\R}}\\
     0 & 1 & \omega \frac{\overline{B}}{A^{\R}}\\
     0 & 0 & 1
    \end{pmatrix}, \; 
    D_6 = \begin{pmatrix}
       \frac{1}{3A^{\R}} & & \\
        & 1 & \\
        & & 3 A^{\R}
    \end{pmatrix}\\
    R_6 &= \begin{pmatrix}
        1 & 0 & 0\\
        -\frac{\overline{B}}{A^{\R}} & 1 & 0\\
        -\omega^2 \frac{|B|^2}{(A^{\R})^2} - \omega^2 s^{\R} - \frac{B}{A^{\R}} & \omega^2 \frac{B}{A^{\R}} & 1
    \end{pmatrix}.
\end{aligned} \label{parametrization of L_6 D_6 R_6 ver.2}
\end{align}
Similarly, we write
\begin{align}
    &\check{Q}_m^{(0,\infty)}= e^{x \theta(\zeta)} Q_m^{(0,\infty)} e^{-x \theta(\zeta)},\quad m\in\frac13 \Z. 
\end{align}


\section{Factorization of Jump Matrices $G_R^{(u)}$} \label{factorization of the jump matrices}

The desired factorization is 
\begin{align*}
    G_R^{(u)} (\zeta) = G_{\tilde{R}}^{(u)} (\zeta) \left[ E^{(u)} (\zeta) \right]^{-1}
\end{align*}
where $E^{(u)} (\zeta) = \O(1)$ and $G_{\tilde{R}}^{(u)} (\zeta) = I + \O(x^{-1})$ as $x \to \infty$. 
Indeed, for $\zeta \in \partial U_{-\overline{\omega}}$, one has
\begin{align}
\begin{aligned}
    &G_R^{(- \overline{\omega})}(\zeta) = 
    \begin{pmatrix}
        1 & 0 & \frac{\widehat{\alpha} \omega}{\zeta + \overline{\omega}}\\
        0 & 1 & 0\\
        \frac{\widehat{\beta}}{\zeta + \overline{\omega}} & 1 & 0
    \end{pmatrix} + \O(x^{-1})
    = G_{\tilde{R}}^{(- \overline{\omega})}(\zeta) \left[E^{(- \overline{\omega})}(\zeta) \right]^{-1},\\
    &E^{(- \overline{\omega})}(\zeta) = \begin{pmatrix}
        1 & 0 & -\frac{\widehat{\alpha} \omega}{\zeta + \overline{\omega}}\\
        0 & 1 & 0\\
        0 & 0 & 1
    \end{pmatrix},\\
    &G_{\tilde{R}}^{(- \overline{\omega})}(\zeta) = \begin{pmatrix}
        1 & 0 & 0\\
        0 & 1 & 0\\
        \frac{\widehat{\beta}}{\zeta + \overline{\omega}} & 0 & 1 - \frac{\widehat{\alpha} \widehat{\beta} \omega}{(\zeta + \overline{\omega})^2}
    \end{pmatrix} + \O(x^{-1}) = I + \O(x^{-1}).
\end{aligned} \label{G^(-omega bar) = E^(-omega bar) G^(-omega bar)}
\end{align}
For $\zeta \in \partial U_{\omega}$, one has
\begin{align}
\begin{aligned}
    &G_R^{(\omega)}(\zeta) = 
    \begin{pmatrix}
        1 & \frac{\widehat{\alpha} \omega}{\zeta - \omega} & 0 \\
        \frac{\widehat{\beta} \omega}{\zeta - \omega} & 1 & 0\\
        0 & 0 & 1
    \end{pmatrix} + \O(x^{-1})
    = G_{\tilde{R}}^{(\omega)}(\zeta)\left[E^{(\omega)}(\zeta)\right]^{-1},\\
    &E^{(\omega)}(\zeta) = \begin{pmatrix}
        1 & -\frac{\widehat{\alpha} \omega}{\zeta - \omega} & 0\\
        0 & 1 & 0\\
        0 & 0 & 1
    \end{pmatrix},\\
    &G_{\tilde{R}}^{(\omega)}(\zeta) = \begin{pmatrix}
        1 & 0 & 0\\
        \frac{\widehat{\beta} \omega}{\zeta - \omega} & 1 - \frac{\widehat{\alpha} \widehat{\beta} \omega^2}{(\zeta - \omega)^2} & 0\\
        0 & 0 & 1
    \end{pmatrix} + \O(x^{-1}) = I + \O(x^{-1}).
\end{aligned} \label{G^(omega) = E^(omega) G^(omega)}
\end{align}
For $\zeta \in \partial U_{-1}$, one has
\begin{align}
\begin{aligned}
    &G_R^{(-1)}(\zeta) = 
    \begin{pmatrix}
        1 & 0 & 0\\
        0 & 1 & \frac{\widehat{\beta} \omega}{\zeta + 1}\\
        0 & \frac{\widehat{\alpha} \omega^2}{\zeta + 1} & 1
    \end{pmatrix}  + \O(x^{-1})
    = G_{\tilde{R}}^{(-1)}(\zeta) \left[ E^{(-1)}(\zeta) \right]^{-1},\\
    &E^{(-1)}(\zeta) = \begin{pmatrix}
        1 & 0 & 0\\
        0 & 1 & 0\\
        0 & -\frac{\widehat{\alpha} \omega^2}{\zeta + 1} & 1
    \end{pmatrix},\\
    &G_{\tilde{R}}^{(-1)}(\zeta) = \begin{pmatrix}
        1 & 0 & 0\\
        0 & 1 - \frac{\widehat{\alpha} \widehat{\beta}}{(\zeta + 1)^2} & \frac{\widehat{\beta} \omega}{\zeta + 1}\\
        0 & 0 & 1
    \end{pmatrix} + \O(x^{-1}) = I + \O(x^{-1}).
\end{aligned} \label{G^(-1) = E^(-1) G^(-1)}
\end{align}
For $\zeta \in \partial U_{\overline{\omega}}$, one has
\begin{align}
\begin{aligned}
    &G_R^{(\overline{\omega})}(\zeta) =
    \begin{pmatrix}
        1 & 0 & \frac{\widehat{\beta} \omega^2}{\zeta - \overline{\omega}}\\
        0 & 1 & 0\\
        \frac{\widehat{\alpha} \omega^2}{\zeta - \overline{\omega}} & 0 & 1
    \end{pmatrix} + \O(x^{-1})
    = G_{\tilde{R}}^{(\overline{\omega})}(\zeta) \left[ E^{(\overline{\omega})}(\zeta) \right]^{-1},\\
    &E^{(\overline{\omega})}(\zeta) = \begin{pmatrix}
        1 & 0 & 0\\
        0 & 1 & 0\\
        -\frac{\widehat{\alpha} \omega^2}{\zeta - \overline{\omega}} & 0 & 1
    \end{pmatrix},\\
    &G_{\tilde{R}}^{(\overline{\omega})}(\zeta) = \begin{pmatrix}
        1 - \frac{\widehat{\alpha} \widehat{\beta} \omega}{(\zeta - \overline{\omega})^2} & 0 & \frac{\widehat{\beta} \omega^2}{\zeta - \overline{\omega}}\\
        0 & 1 & 0\\
        0 & 0 & 1
    \end{pmatrix} + \O(x^{-1}) = I + \O(x^{-1}).
\end{aligned} \label{G^(omega bar) = E^(omega bar) G^(omega bar)}
\end{align}
For $\zeta \in \partial U_{-\omega}$, one has
\begin{align}
\begin{aligned}
    &G_R^{(-\omega)}(\zeta) = 
    \begin{pmatrix}
        1 & \frac{\widehat{\beta} \omega^2}{\zeta + \omega} & 0\\
        \frac{\widehat{\alpha}}{\zeta + \omega} & 1 & 0\\
        0 & 0 & 1
    \end{pmatrix} + \O(x^{-1})
    = G_{\tilde{R}}^{(-\omega)}(\zeta) \left[ E^{(-\omega)}(\zeta) \right]^{-1},\\
    &E^{(-\omega)}(\zeta) = \begin{pmatrix}
        1 & 0 & 0\\
        -\frac{\widehat{\alpha}}{\zeta + \omega} & 1 & 0\\
        0 & 0 & 1
    \end{pmatrix},\\
    &G_{\tilde{R}}^{(-\omega)}(\zeta) = \begin{pmatrix}
        1 - \frac{\widehat{\alpha} \widehat{\beta} \omega^2}{(\zeta + \omega)^2} & \frac{\widehat{\beta} \omega^2}{\zeta + \omega} & 0\\
        0 & 1 & 0\\
        0 & 0 & 1
    \end{pmatrix} + \O(x^{-1}) = I + \O(x^{-1}).
\end{aligned} \label{G^(-omega) = E^(-omega) G^(-omega)}
\end{align}


\section{Alternative Method for Finding $R(0)$}

Finding the unknown matrices $B_1$ and $B_2$ is tedious since we need to compute the inverse of the $6 \times 6$ matrix $M$ (see section \ref{section on asymptotic analysis}). However, we only need $B_2$ to find $R(0) \, (= \Tilde{R}(0))$. More precisely, our ultimate goal is to extract the expression for $e^{-2v_0}$ from it. In this appendix, we propose another approach concerning the last part of the asymptotic analysis.

We first take it for granted that $\hat{R}(\zeta) = I + \O(x^{-1})$ for every $\zeta$ and the Ansatz
\begin{align}
\begin{aligned}  
    \tilde{R}(\zeta) = I &+ \check{A} \,  \diag\left\{\frac{1}{\zeta+\omega},\frac{1}{\zeta+1},\frac{1}{\zeta+\bar{\omega}} \right\}\\
    &+ \check{B} \diag \left\{\frac{1}{\zeta-\bar{\omega}},\frac{1}{\zeta- {\omega}},\frac{1}{\zeta-1}\right\} + \O(x^{-1}), \quad x \to \infty,
\end{aligned}
\end{align}
where $\check{A}$ and $\check{B}$ are both constant matrices.
In other words, one can write $\tilde{R}(\zeta)$ by
\begin{align}
    \tilde{R}(\zeta) &= (\e_1, \e_2, \e_3) + \left( \frac{1}{\zeta + \omega} \check{A}^{(1)}, \frac{1}{\zeta + 1} \check{A}^{(2)}, \frac{1}{\zeta + \omega^2} \check{A}^{(3)} \right)\\
    &+ \left( \frac{1}{\zeta - \omega^2} \check{B}^{(1)}, \frac{1}{\zeta - \omega} \check{B}^{(2)}, \frac{1}{\zeta - 1} \check{B}^{(3)} \right) + \O(x^{-1})
\end{align}
where $\check{A}^{(i)}$ and $\check{B}^{(i)}$ are the $i$-th column vectors of matrices $\check{A}$ and $\check{B}$ respecitvely.

The residue relations \eqref{residue relation 1} -- \eqref{residue relation 6} for $\tilde{R}(\zeta)$ give
\begin{align*}
    &\check{A}^{(1)} = -\hat{\alpha}(1,x) \left(\e_2 + \dfrac{1}{1-\omega} \check{A}^{(2)} + \dfrac{1}{-2\omega} \check{B}^{(2)} \right) + \O(x^{-1}),\\
    &\check{A}^{(2)} = -\hat{\alpha}(1,x) \omega^2\left( \e_3 + \dfrac{1}{\omega^2-1} \check{A}^{(3)} + \dfrac{1}{-2} \check{B}^{(3)} \right) + \O(x^{-1}),\\
    &\check{A}^{(3)} = -\hat{\alpha}(1,x) \omega \left( \e_1 + \dfrac{1}{\omega-\omega^2}\check{A}^{(1)} + \dfrac{1}{-2\omega^2} \check{B}^{(1)} \right) + \O(x^{-1}),\\
    &\check{B}^{(1)} = -\hat{\alpha}(1,x) \omega^2 \left( \e_3 + \dfrac{1}{2\omega^2} \check{A}^{(3)} + \dfrac{1}{\omega^2-1} \check{B}^{(3)} \right) + \O(x^{-1}),\\
    &\check{B}^{(2)} = -\hat{\alpha}(1, x) \omega \left( \e_1 + \dfrac{1}{2\omega} \check{A}^{(1)} + \dfrac{1}{\omega-\omega^2} \check{B}^{(1)} \right) + \O(x^{-1}),\\
    &\Check{B}^{(3)} = -\hat{\alpha}(1,x) \left( \e_2 + \dfrac{1}{2} \check{A}^{(2)} +\dfrac{1}{1- {\omega}} \check{B}^{(2)} \right) + \O(x^{-1}).
\end{align*}
The relations above can be written in the form:
\begin{align}
    M_1 \mathbf{w}_1 = \mathbf{w}_2
\end{align}
where 
\begin{align}
&M_1 = 
\begin{pmatrix} 
   1&\frac{\hat{\alpha}(1,x)}{1-\omega } & 0 & 0 & \frac{-\hat{\alpha}(1,x)}{2\omega} & 0\\
   0 & 1 & \frac{\hat{\alpha}(1,x) \omega^2}{\omega^2 - 1} & 0 & 0 & \frac{-\hat{\alpha}(1,x) \omega^2}{2}\\
   \frac{\hat{\alpha}(1,x) \omega}{\omega - \omega^2} & 0 & 1 & \frac{-\hat{\alpha}(1,x)}{2\omega} & 0 & 0\\
   0 & 0 & \frac{\hat{\alpha}(1,x)}{2} & 1 & 0 & \frac{\hat{\alpha}(1,x) \omega^2}{\omega^2 - 1}\\
   \frac{\hat{\alpha}(1,x)}{2} & 0 & 0 & \frac{\hat{\alpha}(1,x) \omega}{\omega - \omega^2} & 1 & 0 \\
   0 & \frac{\hat{\alpha}(1,x)}{2} & 0 & 0 & \frac{\hat{\alpha}(1,x)}{1-\omega} & 1
\end{pmatrix}\\
&
\mathbf{w}_1 = \begin{pmatrix} 
    \check{A}^{(1)}\\
    \check{A}^{(2)}\\
    \check{A}^{(3)}\\
    \check{B}^{(1)}\\
    \check{B}^{(2)}\\
    \check{B}^{(3)}
\end{pmatrix}, \quad
\mathbf{w}_2 = -\hat{\alpha}(1,x)
    \begin{pmatrix}
        \e_2\\
        \omega^2 \e_3\\
        \omega \e_1\\
        \omega^2 \e_3\\
        \omega \e_1\\
        \e_2
    \end{pmatrix} + \O(x^{-1}).
\end{align}

So we need to solve only one 6 by 6 system, which looks simpler than the three systems \eqref{residue relations III} for the leading terms of $B_1$ and $B_2$.
Observe that 
\begin{align}
    \tilde{R}(0)=I + \omega^2\check{A}- \omega \check{B} + \O(x^{-1}).
\end{align}
Moreover, from \eqref{R(0) asymptotic 2}, we see that the sum of the first column of $\tilde{R}(0)$ gives $-e^{-2v_0}$, so we are particularly interested in
\begin{align}
    \tilde{R}^{(1)}(0) = \e_1 + \omega^2 \check{A}^{(1)} -  \omega \check{B}^{(1)} + \O(x^{-1}).
    \label{main interest maksim}
\end{align}
To obtain the desired asymptotics \eqref{e^-2v = fraction of sin} from \eqref{main interest maksim}, one can proceed as follows:
\noindent First, multiply by $\omega$ in the first column and by $\omega^2$ in the fourth column of $M_1$ to have
\begin{align*}
\left(
\begin{smallmatrix} 
  {\color{webbrown} \omega} & \frac{\hat{\alpha}(1,x)}{1-\omega} & 0 & 0 & \frac{-\hat{\alpha}(1,x)}{2\omega} & 0\\
  0 & 1 & \frac{\hat{\alpha}(1,x) \omega^2}{\omega^2-1} & 0 & 0 & \frac{-\hat{\alpha}(1,x) \omega^2}{2}\\
  {\color{webbrown}\frac{\hat{\alpha}(1,x) \omega^2}{\omega - \omega^2}} & 0 & 1 &{\color{webbrown}\frac{\hat{\alpha} (1,x) \omega}{2}} & 0 & 0\\
  0 & 0 & \frac{\hat{\alpha}(1,x)}{2} &{\color{webbrown}-\omega^2} & 0 & \frac{\hat{\alpha}(1,x) \omega^2}{\omega^2-1}\\
   {\color{webbrown}\frac{\hat{\alpha} (1,x) \omega}{2}} & 0 & 0 &{\color{webbrown}\frac{-\hat{\alpha}(1,x)}{\omega-\omega^2}} & 1 & 0\\
   0 & \frac{\hat{\alpha}(1,x)}{2} & 0 & 0 & \frac{\hat{\alpha}(1,x)}{1-\omega} & 1
   \end{smallmatrix}
   \right) 
   \left(
   \begin{smallmatrix}
   {\color{webbrown}\omega^2}\check{A}^{(1)}\\
   \check{A}^{(2)}\\
   \check{A}^{(3)}\\
   {\color{webbrown} -\omega\check{B}^{(1)}}\\
   \check{B}^{(2)}\\
   \check{B}^{(3)}
   \end{smallmatrix} 
   \right) = \mathbf{w}_2.
\end{align*}
\noindent Then, subtract the first column of the resulting matrix from the fourth column to have
\begin{align}
\quad 
\left(
\begin{smallmatrix} 
  {\omega} & \frac{\hat{\alpha}(1,x)}{1-\omega} & 0 &{\color{webbrown}-\omega} & \frac{-\hat{\alpha}(1,x)}{2\omega} & 0\\
  0 & 1 & \frac{\hat{\alpha}(1,x) \omega^2}{\omega^2-1} & 0 & 0 & \frac{-\hat{\alpha}(1,x) \omega^2}{2}\\
  \frac{\hat{\alpha}(1,x) \omega^2}{\omega-\omega^2} & 0 & 1 & \frac{\hat{\alpha}(1,x) \omega}{2} -{\color{webbrown}\frac{\hat{\alpha} (1,x) \omega^2}{\omega-\omega^2}} & 0 & 0\\
  0 & 0 & \frac{\hat{\alpha}(1,x)}{2} & -\omega^2 & 0 & \frac{\hat{\alpha}(1,x) \omega^2}{\omega^2-1}\\
   \frac{\hat{\alpha}(1,x) \omega}{2} & 0 & 0 & \frac{-\hat{\alpha}(1,x)}{\omega-\omega^2} -{\color{webbrown}\frac{\hat{\alpha}(1,x) \omega}{2}} & 1 & 0\\
   0 & \frac{\hat{\alpha}(1,x)}{2} & 0 & 0 & \frac{\hat{\alpha}(1,x)}{1-\omega} & 1
   \end{smallmatrix}
   \right)
   \left(
   \begin{smallmatrix} 
    \omega^2 \check{A}^{(1)}{\color{webbrown} -\omega\check{B}^{(1)}}\\
    \check{A}^{(2)}\\
    \check{A}^{(3)}\\
    -\omega\check{B}^{(1)}\\
    \check{B}^{(2)}\\
    \check{B}^{(3)}
    \end{smallmatrix}
    \right)
    = \mathbf{w}_2.
    \label{M2 vec = -alpha vec}
\end{align}
We shall denote the matrix on the left-hand-side of \eqref{M2 vec = -alpha vec} by $M_2$, then
\begin{align}
    \det M_2 = -\det M_1. \label{det M2 = - det M1}
\end{align}
Since $\check{A}^{(1)}$ and $\check{B}^{(1)}$ are both vectors in $\R^3$, the solution of system \eqref{M2 vec = -alpha vec} must be of the form
\begin{align}
    \omega^2 \check{A}^{(1)} - \omega \check{B}^{(1)} = X_{11} \e_1 + X_{12} \e_2 + X_{13} \e_3,
    \label{omega^2 A - omega B}
\end{align}
with some $X_{1j} \in \C$. 
From \eqref{main interest maksim} and \eqref{omega^2 A - omega B}, 
\begin{align}
-e^{-2v_0} &= \text{the sum of all entries of $\Tilde{R}^{(1)}(0)$}\\
&= 1 + X_{11} + X_{12} + X_{13} + \O(x^{-1}).
\end{align}
The sum $X_{11} + X_{12} + X_{13}$ can be obtained when we solve \eqref{M2 vec = -alpha vec} with replacement of $\e_{1}$, $\e_2$, and $\e_3$ in $\mathbf{w}_2$ on the right-hand-side by $1$. Namely,
\begin{align}\label{cramer's rule}
&-e^{-2v_0} = 1 - \frac{
\det
\begin{pNiceArray}{c|ccccc}[small]
   \Block{6-1}{\mathbf{w_3}} & \frac{\hat{\alpha}(1,x)}{1-\omega} & 0 & -\omega & \frac{-\hat{\alpha}(1,x)}{2\omega} & 0\\
   \vline & 1 & \frac{\hat{\alpha}(1,x)\omega^2}{\omega^2-1} & 0 & 0 & \frac{-\hat{\alpha}(1,x)\omega^2}{2}\\
   & 0 & 1 & \frac{\hat{\alpha}(1,x)\omega}{2} - \frac{\hat{\alpha}(1,x)\omega^2}{\omega-\omega^2} & 0 & 0\\
   & 0 & \frac{\hat{\alpha}(1,x)}{2}&{ -\omega^2} & 0 & \frac{\hat{\alpha}(1,x)\omega^2}{\omega^2-1}\\
   & 0 & 0 & \frac{-\hat{\alpha}(1,x)}{\omega-\omega^2} - \frac{\hat{\alpha}(1,x)\omega}{2} & 1 & 0 \\
   & \frac{\hat{\alpha}(1,x)}{2} & 0 & 0 & \frac{\hat{\alpha}(1,x)}{1-\omega} & 1
\end{pNiceArray}
}{\det M_1}
+ \O(x^{-1})
\end{align}
where we set
\begin{align}
    \mathbf{w}_3 = -\hat{\alpha}(1,x)
    \begin{pmatrix}
        1\\
        \omega^2\\
        \omega\\
        \omega^2\\
        \omega\\
        1
    \end{pmatrix} + \O(x^{-1})
\end{align}
and then used Cramer's rule and \eqref{det M2 = - det M1} in the computation.
If we introduce $\phi$ as in \eqref{defn of arg = phi}, then one can obtain
\begin{align}
    \det M_1 = 2 i e^{3i\phi} \left(9\sin \left(\phi-\frac{\pi }{3} \right) + 8 -\sin (3\phi) \right),
\end{align}
while the numerator determinant on the right-hand-side of \eqref{cramer's rule} becomes
\begin{align}
    6 i e^{3 i \phi} \left( 5 + 2 \sin^2 \phi + 4 \sin \left(\phi-\frac{\pi}{3} \right) + \sqrt3 \sin (2\phi)\right) + o(1).
\end{align}
Thus, it follows that
\begin{align}
 -e^{-2v_0}=1-\dfrac{ 3\left(5 + 2\sin^2\phi +4 \sin \left( \phi-\frac{\pi}{3} \right) + \sqrt3\sin (2\phi) \right)}{9\sin \left( \phi-\frac{\pi}{3} \right) + 8 - \sin (3\phi)} + \O(x^{-1})
\end{align}
which after simplification gives 
\begin{align}
    e^{-2v_0}= \dfrac{2-\sin(\phi-\frac{\pi}{3})}{1+\sin(\phi-\frac{\pi}{3})} + \O(x^{-1}),
\end{align}
in agreement with \eqref{e^-2v = fraction of sin}.